\documentclass[onecolumn,journal,12pt,a4paper,twoside,draftcls]{IEEEtran}
\usepackage{indentfirst}
\usepackage{amsmath}
\usepackage{amsthm}
\usepackage{amsfonts,color,xcolor}
\usepackage[noadjust]{cite}
\usepackage[pdftex]{graphicx}
\usepackage{mathtools}
\usepackage{comment}
\usepackage{lipsum}
\usepackage{mathabx}
\usepackage{mathrsfs,bm,bbm}
\mathtoolsset{showonlyrefs}
\IEEEoverridecommandlockouts
\ifCLASSOPTIONcompsoc
\usepackage[caption=false,font=normalsize,labelfont=sf,textfont=sf]{subfig}
\else
\usepackage[caption=false,font=footnotesize]{subfig}
\fi

\usepackage[russian,USenglish]{babel}

\newcommand{\rev}{\textcolor{black}}

\newcommand{\figr}{Fig.~}
\newcommand{\secr}{Sec.~}

\usepackage{etoolbox}
\usepackage{tikz}
\newrobustcmd*{\mycircle}[1]{\tikz{\filldraw[draw=#1,fill=#1] (0,-0.3) circle [radius=0.08cm];}}
\newrobustcmd*{\squareA}[1]{\tikz{\filldraw[draw=#1,fill=#1] (0,-0)
rectangle (0.1cm,0.14cm);}}
\usetikzlibrary{shapes,arrows}
\usetikzlibrary{positioning,arrows.meta}
\let\OLDthebibliography\thebibliography
  \renewcommand\thebibliography[1]{
  \OLDthebibliography{#1}
  \enlargethispage{2mm}
  \vspace{-0.8mm}
  \setlength{\parskip}{0pt}
  \setlength{\itemsep}{-0.2pt}
}

\usepackage{acro}
\DeclareAcronym{AWGN}{short = AWGN ,long = additive white gaussian noise}
\DeclareAcronym{AoI}{short = AoI ,long = age of information}
\DeclareAcronym{AoII}{short = AoII ,long = age of incorrect information}

\DeclareAcronym{CDF}{short = CDF ,long = cumulative distribution function}
\DeclareAcronym{CRA}{short = CRA ,long = contention resolution ALOHA}
\DeclareAcronym{CRDSA}{short = CRDSA ,long = contention resolution diversity slotted ALOHA}
\DeclareAcronym{CSA}{short = CSA ,long = coded slotted ALOHA}
\DeclareAcronym{C-RAN}{short = C-RAN ,long = cloud radio access network}
\DeclareAcronym{DAMA}{short = DAMA ,long = demand assigned multiple access}
\DeclareAcronym{DSA}{short = DSA ,long = diversity slotted ALOHA}
\DeclareAcronym{eMBB}{short = eMBB ,long = enhanced mobile broadband}
\DeclareAcronym{FEC}{short = FEC ,long = forward error correction}
\DeclareAcronym{GEO}{short = GEO ,long = geostationary orbit}
\DeclareAcronym{GF}{short = GF ,long = generating function}
\DeclareAcronym{HMM}{short = HMM ,long = hidden Markov model}
\DeclareAcronym{IC}{short = IC ,long = interference cancellation}
\DeclareAcronym{IoT}{short = IoT ,long = Internet of Things}
\DeclareAcronym{IRSA}{short = IRSA ,long = irregular repetition slotted ALOHA}
\DeclareAcronym{KPI}{short = KPI, long = key performance indicator}
\DeclareAcronym{LEO}{short = LEO ,long = low Earth orbit}
\DeclareAcronym{MAC}{short = MAC ,long = medium access}
\DeclareAcronym{MAP}{short = MAP, long = maximum a posteriori}
\DeclareAcronym{mMTC}{short = mMTC ,long = massive machine-type communications}
\DeclareAcronym{MC}{short = MC ,long = Markov chain}
\DeclareAcronym{NTN}{short = NTN, long = non-terrestrial network}
\DeclareAcronym{PDF}{short = PDF ,long = probability density function}
\DeclareAcronym{PER}{short = PER ,long = packet error rate}
\DeclareAcronym{PLR}{short = PLR ,long = packet loss rate}
\DeclareAcronym{PMF}{short = PMF ,long = probability mass function}
\DeclareAcronym{RA}{short = RA ,long = random access}
\DeclareAcronym{rv}{short = r.v. ,long = random variable}
\DeclareAcronym{SA}{short = SA , long = slotted ALOHA}
\DeclareAcronym{SIC}{short = SIC ,long = successive interference cancellation}
\DeclareAcronym{SNR}{short = SNR ,long = signal-to-noise ratio}
\DeclareAcronym{SFG}{short = SFG ,long = signal flow graph}
\DeclareAcronym{TDM}{short = TDM ,long = time division multiplexing}
\DeclareAcronym{VoI}{short = VoI, long = value of information}

\newtheorem{lemma}{Lemma}
\newtheorem{remark}{Remark}
\newtheorem{example}{Example}

\newcommand{\pb}{\ensuremath{\overline{p}}}
\newcommand{\pr}{\ensuremath{P}}

\newcommand{\given}{\, | \,}
\newcommand{\givenS}{\ensuremath{\vert}}
\newcommand{\norm}[2]{\left \lVert #1 \right \rVert_{#2}}
\newcommand{\dens}{\ensuremath{\rho}}
\newcommand{\area}{\ensuremath{\mathcal A}}
\newcommand{\rad}{\ensuremath{R}}
\newcommand{\radMax}{\ensuremath{\rad_{m}}}
\newcommand{\powLaw}{\ensuremath{\alpha}}

\newcommand{\Mcn}{\ensuremath{X_n}}
\newcommand{\McA}{\ensuremath{\mathcal{X}}}
\newcommand{\TransMat}{\ensuremath{T}}
\newcommand{\Rc}{\ensuremath{Y}}
\newcommand{\rc}{\ensuremath{y}}
\newcommand{\Rcn}{\ensuremath{\Rc_n}}
\newcommand{\rcn}{\ensuremath{\rc_n}}
\newcommand{\Rcnvec}{\ensuremath{\Rc^n}}
\newcommand{\rcnvec}{\ensuremath{\rc^n}}
\newcommand{\Distn}{\ensuremath{D_n}}
\newcommand{\distn}{\ensuremath{d_n}}
\newcommand{\dist}{\ensuremath{d}}
\newcommand{\mc}{\ensuremath{x}}
\newcommand{\mcn}{\ensuremath{\mc_n}}

\newcommand{\Readn}{\ensuremath{{Z_n}}}
\newcommand{\readn}{\ensuremath{z_n}}
\newcommand{\lastRcn}{\ensuremath{{W_n}}}
\newcommand{\lastrcn}{\ensuremath{w_n}}
\newcommand{\lastRc}{\ensuremath{{W}}}
\newcommand{\lastrc}{\ensuremath{w}}

\newcommand{\pTx}{\ensuremath{\zeta}}
\newcommand{\ps}{\ensuremath{p_s}}
\newcommand{\peras}{\ensuremath{p_e}}

\newcommand{\asymm}{\ensuremath{\eta}}

\newcommand{\nodes}{\ensuremath{m}}

\newcommand{\Agen}{\ensuremath{\Delta_n}}
\newcommand{\agen}{\ensuremath{\delta_n}}
\newcommand{\avgAge}{\ensuremath{\bar{\Delta}}}
\newcommand{\ent}{\ensuremath{H}}
\newcommand{\condentforg}{\ensuremath{h}_f}
\newcommand{\condent}{\ensuremath{h}}
\newcommand{\entFun}{\ensuremath{\tilde H}}

\newtheorem{theorem}{Theorem}

\begin{document}

\title{A Spatio-Temporal Model for Information Freshness in Massive Random Access}
\author{Andrea~Munari,~\IEEEmembership{Senior Member,~IEEE}, Alessandro Buratto,~\IEEEmembership{Student Member,~IEEE}, Federico Chiariotti,~\IEEEmembership{Senior Member,~IEEE}, Leonardo Badia,~\IEEEmembership{Senior~Member,~IEEE}, and Petar Popovski,~\IEEEmembership{Fellow,~IEEE}
\thanks{A. Munari is with the Institute of Communications and Navigation, German Aerospace Center (DLR), We\ss ling, Germany. email: andrea.munari@dlr.de. A. Buratto, F. Chiariotti, and L. Badia are with the Dept. of Information Engineering, University of Padova, Italy. email: \{alessandro.buratto.1@studenti.unipd.it, federico.chiariotti@unipd.it, leonardo.badia@unipd.it\}. P. Popovski is with the Dept. of Electronic Systems, Aalborg University, Denmark. email: petarp@es.aau.dk.}
\thanks{This work was supported in part by the Italian National Recovery and Resilience Plan (NRRP), as part of the RESTART partnership (PE0000001), under the European Union NextGenerationEU Project, by the the German Federal Ministry of Research, Technology, and Space (BMFTR) with the xG-RIC project as part of the research program Communication Systems 
“Souverän. Digital. Vernetzt.” (grant number 16KIS2429K), and by the Villum Investigator Grant ``WATER'' from the Velux Foundations, Denmark.}
\thanks{A preliminary version of this work \cite{Munari25_SPAWC} was presented at IEEE SPAWC, Surrey (UK), 7--10 July, 2025.}
\vspace{-1.5cm}
}
\date{}
\thispagestyle{empty}
\maketitle

\begin{abstract}
\vspace{-0.2cm}
Massive connectivity, a key building block of 5G, is expected to play an important role in the next generation of wireless systems, 
 and its expected requirements are being revolutionized through the modeling of the information dynamics related to the vast numbers of Internet of things (IoT) devices. 
Motivated by this, the present paper introduces a model \rev{that captures} the spatio-temporal \rev{nature} of freshness of information sent via random access channel policies from an extremely large set of IoT devices \rev{via simple scalar parameters, i.e., the probability of success and accuracy of received updates}. There are many information freshness metrics, starting from the age of information (AoI), all of which are proxies for the actual application performance, characterized over the temporal dimension. Our model adds the spatial dimension to this picture, observing that sensors distributed over the same area may have a strong correlation, and information from multiple close-by sensors may improve the overall accuracy of the receiver. We focus on characterizing the uncertainty of the receiver, expressed through the conditional entropy, considering a network of partially reliable, spatially distributed sensors observing the same process and reporting their measurements over a slotted ALOHA channel. We consider a simple forgetful receiver and a more complete model which accounts for the full history of past observations, deriving their performance, and optimizing the transmission probability of nodes to minimize the expected uncertainty.
\end{abstract}


\section{Introduction}
\label{sec:intro}

The rapid development of the \ac{IoT} over the past decade has led to the deployment of tens of billions of distributed sensors,  enabling a wide range of applications, from predictive maintenance in manufacturing plants to environmental monitoring~\cite{ericsson2024mobility, kalor2024wireless}. However, managing a massive number of low-power sensors, with limited communication capabilities and severe computational and energy constraints due to their small batteries, poses several challenges captured by novel key performance indicators and constraints~\cite{guo2021enabling}, which have shaped and oriented recent research on medium access protocol design.

In 5G, massive machine-type communication has already been recognized as a cornerstone service category, with random access protocols playing a central role in enabling scalable connectivity under sporadic traffic patterns. Yet, the requirements of 6G push this problem to a new level. Beyond sheer connectivity, 6G \ac{IoT} scenarios are tightly linked to the concept of goal-oriented communications~\cite{gunduz2023timely}, where communication protocols shall be designed and optimized to convey information that serves a concrete task, e.g., monitoring, and where information freshness and accuracy become as important as throughput.

In this context, \ac{AoI} has emerged as a key metric of timeliness~\cite{Kaul11_SECON}: unlike traditional network latency, which focuses on a single packet, \ac{AoI} measures the freshness of the most recent update available to the receiver, thus acting as a much closer proxy for the actual misalignment between the real condition of the measured process and the best available estimate. Since its introduction in 2011, \ac{AoI} has been an active research field, and a significant amount of scientific literature has been dedicated to its characterization~\cite{Yates19_TIT} and its use in network optimization~\cite{Modiano18_AoI}. Several extensions of \ac{AoI} have also been developed, providing even closer proxies of the application performance, at the cost of higher complexity and application-specific definitions. The \ac{AoII}~\cite{Ephremides19_AoII} and \ac{VoI}~\cite{Kellerer19} are particularly interesting, as they consider the estimation error directly in their calculation, penalizing scenarios in which the estimate is misaligned with the actual state of the system.

However, while the temporal dimension of information freshness has been thoroughly explored and characterized, the spatially distributed nature of \ac{IoT} networks has been mostly neglected so far, or only considered as part of the medium access model~\cite{yang2021spatiotemporal}. In real \ac{IoT} networks, the physical distance of a sensor from the target event to be measured is one of the factors determining the quality of individual measurements, and fresher, lower-quality measurements may be less informative than older, better measurements. Consequently, metrics such as \ac{AoI} become less useful as proxies for the actual knowledge accuracy at the receiver.

Corrections to the definitions of \ac{AoI} or \ac{VoI} that account for this phenomenon have been proposed~\cite{tong2022age,zancanaro2023modeling,fidler20242d}, but these are still indirect measurements of the knowledge uncertainty. We argue that the entropy of the \emph{a posteriori} distribution of the state, conditioned on the information available to the receiver~\cite{luo_TIT2009,Cocco23_JSAIT}, provides a direct characterization of the performance of the monitoring system, without the need for such proxies, and is thus the ideal target for massive access optimization. This work is inspired by the use of entropy as an estimation and control theoretical metric~\cite{saridis2002entropy}, which has shown significant advantages over traditional metrics such as minimum mean square error under non-Gaussian errors~\cite{chen2019minimum}. Additionally, the receiver entropy can capture more than strictly spatial and temporal relations, as it is tightly coupled with system dynamics and structural factors.

\subsection{Motivating Example}

\begin{figure}
    \centering
    \definecolor{forestgreen}{RGB}{49,118,63}
\definecolor{darkbrown}{RGB}{141,110,56}
\tikzset{pics/.cd,
  sensor/.style={
    code={
        \draw[pic actions] (-0.04,-0.115) rectangle (0.04,0.025);
        \draw (0,0.025) -- (0,0.085);
        \draw (0,0.085) -- (0.03,0.115);
        \draw (0,0.085) -- (-0.03,0.115);
    } 
  },
  tree/.style={
    code={
        \draw[fill=darkbrown] (-0.03,-0.1) rectangle (0.03,-0.3);
        \node [cloud, inner sep=2pt,draw,minimum width = 0.25cm,fill=forestgreen, aspect=2] {};
    } 
  },
  satellite/.style={
    code={
        \draw (0.07,-0.07) -- (0.17,-0.17);
        \draw[pic actions,rotate around={45:(0.27,0.27)}] (0.095,0.17) rectangle (0.445,0.37);
        \draw (0.38,0.24) -- (0.24,0.38);
        \draw (0.3,0.16) -- (0.16,0.3);
        \draw (0.07,0.07) -- (0.39,0.39);
        \draw[pic actions,rotate around={45:(-0.27,-0.27)}] (-0.095,-0.17) rectangle (-0.445,-0.37);
        \draw (-0.38,-0.24) -- (-0.24,-0.38);
        \draw (-0.3,-0.16) -- (-0.16,-0.3);
        \draw (-0.07,-0.07) -- (-0.39,-0.39);        
        \draw[pic actions] circle[radius=1mm];
        \draw (0,0) -- (-0.17,0.17);
    } 
  }
}

\begin{tikzpicture}[scale=1.1, transform shape]

\draw (0,0)[black,fill=blue!15] ellipse (2.4cm and 1.5cm);
\draw (0,0)[black,dashed,fill=blue!30] ellipse (1.6cm and 1cm);
\draw (0,0)[black,dashed,fill=blue!45] ellipse (0.8cm and 0.5cm);
\draw (2.165,0.66) -- (0,3);
\draw (-2.165,0.66) -- (0,3);

\draw (0,0.3) pic[draw]{tree};
\draw (-0.25,3.25) pic[draw,fill=gray!20]{satellite};

\draw (2,0) pic[draw,fill=gray!20]{sensor};
\draw (0.9,-0.2) pic[draw,fill=gray!20]{sensor};
\draw (0.6,0.85) pic[draw,fill=gray!20]{sensor};
\draw (0.1,1) pic[draw,fill=gray!20]{sensor};
\draw (-0.1,-1.3) pic[draw,fill=gray!20]{sensor};
\draw (-1.1,-0.15) pic[draw,fill=gray!20]{sensor};
\draw (-2,-0.58) pic[draw,fill=gray!20]{sensor};
\draw (1.5,0.75) pic[draw,fill=gray!20]{sensor};
\draw (-1.6,0.95) pic[draw,fill=gray!20]{sensor};
\draw (1.75,0.15) pic[draw,fill=gray!20]{sensor};
\draw (-1.55,-0.45) pic[draw,fill=gray!20]{sensor};
\draw (-1.25,0.65) pic[draw,fill=gray!20]{sensor};
\draw (-0.25,0.05) pic[draw,fill=gray!20]{sensor};


\draw [<->] (0,0) -- (-1.66,-1.1) node [pos=0.66,above] {\tiny $\radMax$};
\draw [<->] (0,0) -- (0,-0.5) node [pos=0.5,right] {\tiny $\rad$};
\node at (0.5,0.15) {\tiny $\area_0$};
\node at (1.1,0.3) {\tiny $\area_1$};
\node at (1.8,0.45) {\tiny $\area_2$};

\end{tikzpicture}
    \caption{A schematic of the considered example scenario.}
    \label{fig:example}
    \vspace{-1em}
\end{figure}
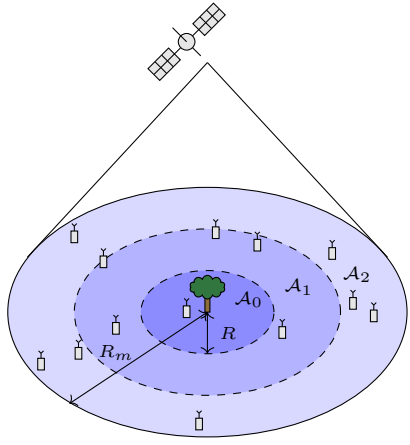

Consider the scenario represented in Fig.~\ref{fig:example}, in which a phenomenon of interest (represented by the tree) is observed by a set of spatially distributed sensors that can perform an uplink transmission towards an \ac{IoT} gateway. This could be instantiated by a terrestrial base station, or by a non-terrestrial receiver, e.g., a \ac{LEO} satellite in the example. Closer sensors will be able to observe the process with a high accuracy, while farther ones will still obtain measurements, albeit with a reduced quality. Sensors beyond a certain distance will not observe the process at all. The IoT application, fed with the data collected by the gateway, monitors the system through sensor updates. 

If observations can be polled at will, 
picking the closest sensor is a natural choice, so as to maximize reliability. 
However, such an approach is hardly viable in practical \ac{IoT} systems, as \textit{(1)} it would quickly deplete the energy of such node, and \textit{(2)} sensors might be able to obtain and report new measurements only sporadically due to duty-cycle and normative constraints. On top of this, data collection is likely to be performed relying on random access protocols over a shared channel, as
scheduled schemes often have signaling requirements that might prove too costly for low-power sensors.
In addition, data collection can be hindered by wireless channel impairments, which can prevent communication. Updates from sensors placed farther away then become beneficial: although they may have a higher error rate, they still provide new information to the gateway, which can be integrated to the current belief over the state of the process. 

This setting presents two major challenges. The first is to determine the accuracy of the knowledge at the receiver's end and its uncertainty over the state of the system. Metrics derived from the \ac{AoI} do not fully capture the scenario, or need significant adjustments to account for the sensors' different levels of accuracy. The second is medium access optimization, as the difference in performance between a simple scheme in which sensors within a given radius transmit with a set rate and a more complex one in which sensors closer to the event of interest transmit more often is an open question. It is possible that preventing farther nodes from transmitting, reducing the frequency of received updates, i.e., the system throughput, but increasing their quality, may have a beneficial effect on the receiver uncertainty.
Furthermore, if sensors have some knowledge of their position, it might be useful to tag messages with the distance, but this incurs additional overhead. A complete characterization of the system also allows for optimization, improved performance, and reduction of unnecessary signaling. We argue that understanding the trade-offs that emerge in these settings is paramount to derive solid design principles of \ac{IoT} massive connectivity in 6G~\cite{kalor2024wireless}, and we aim to provide useful insights in this paper. As we show in the following, this does not impose a complexity burden on the system, but requires a careful optimization of key parameters such as the coverage radius. Notably, simple configurations in which all nodes transmit with the same probability are proven to minimize the uncertainty.

\subsection{Related Work}

As the focus of the research community has moved from characterizing freshness to a more direct consideration of application performance in remote monitoring, several studies have analyzed the effect of different transmission policies in this scenario. In particular, \ac{AoII} has been the subject of a significant amount of work, as it combines the freshness and timeliness properties of \ac{AoI} with the consideration of the mismatch between the physical process and the beliefs of the receiver. In a single transmitter scenario, policies that can optimize \ac{AoII} under rate limitations and communication delays have been designed for slotted~\cite{chen2025preempting} and continuous time~\cite{cosandal2025multi}. In both cases, the optimal transmission policies are based on state-dependent thresholds, which depend on the transition probabilities of the Markov chain that describes the monitored process. A recent work~\cite{cosandal2024joint} has also considered pull-based strategies, i.e., systems in which a receiver might poll one of several sensors, finding that a characterization of the belief over the state, along with the system \ac{AoI}, represent a sufficient statistic to obtain the \ac{AoII}-optimal policy.

Another possibility that has been explored in the past few years has been the direct consideration of the estimator error at the receiver. Naturally, this requires the specification of the estimator and loss function, as different errors might have different impacts on performance. The specification of a cost for false alarms and missed detections was considered for an anomaly reporting system in~\cite{luo2024minimizing}, and the effect of different policies in the same scenario was analyzed in~\cite{salimnejad2023state} using a hold estimator that maintains the last received value as its best estimate of the current state. A more advanced class of estimators was proposed in~\cite{liyanaarachchi2025structured}: as the complexity of a \ac{MAP} approach makes analysis difficult, and martingale-based solutions that exploit the Markov property only represent a limited set of systems, structured estimators that can still prove tractable while representing a wider class of systems are useful in the characterization of this problem.
Finally, \cite{talli2025pragmatic} considered a joint design of communication and estimation policies by modeling the problem as a two-agent system, showing that a joint design can significantly outperform fixed estimators and pull-based schemes.

All the aforementioned works consider systems in which sensors have a perfect observation of the process, and the only bottleneck is communication. Imperfect transmitters have been considered in~\cite{holm2023goal}, which proposed a pull-based reinforcement learning-based scheduler for a receiver that implements a Kalman filter and decides which sensor to poll, and in~\cite{zakeri2025semantic}, which defined a sampling cost for the individual sensor in a push-based system, effectively reducing the observability of the process.

In this work, we are interested in a special case of imperfect observation, in which the spatial distribution of the sensor network is crucial to determine the accuracy of information. While older works only considered the spatial dimension's effect on medium access~\cite{yang2021spatiotemporal}, a recent paper~\cite{tong2022age} proposed a spatio-temporal definition of \ac{AoI}, which models spatial correlation as a proportional decrease of the \ac{AoI} of a sensor if another close-by sensor transmits. A similar definition has been given in~\cite{zancanaro2023modeling}, considering the sum of \ac{AoI} contributions from multiple sensors, weighted by their proximity to the sensor to be scheduled. Another recent work~\cite{fidler20242d} applied a distance-based penalty in determining the freshness for a given point in space, choosing the freshest sample, which might not be from the sensor that is physically closest, depending on its relative age.

However, adapting \ac{AoI} presents some drawbacks, since its value as a proxy for the actual monitoring performance requires a precise modeling of the sensor network and the effect of distance on the accuracy of the observation. 
The entropy of full-sequence~\cite{Cocco23_JSAIT} and forgetful~\cite{Munari25_ISIT} receivers was considered in our previous work under different transmission policies. Unlike the present contribution, these papers consider a case in which observations are perfectly reliable and each node independently monitors a different process.
While the analytical tools are similar to the ones used in this paper,~\cite{Cocco23_JSAIT,Munari25_ISIT} relied on the renewal property of successful transmissions, which does not hold under the present settings, requiring a more complex characterization.
The entropy rate at the receiver was tackled in~\cite{luo_TIT2009}, which considered the tracking of a binary Markov process through a binary channel: this is a complementary scenario with respect to~\cite{Cocco23_JSAIT,Munari25_ISIT}, as it does not include medium access aspects, but focuses on a single node transmitting at every slot over an unreliable channel that may yield incorrect observations. The present work combines both the effect of multiple nodes attempting to access the channel and of their observations being unreliable following a spatial model, requiring a more complex analysis.
The receiver entropy may prove more meaningful than proxy metrics~\cite{rezaeianPercom2007}, as it directly represents the performance of the monitoring   application~\cite{rezaeianArxiv2006}. Particularly relevant to the present paper is the uncertainty of information (UoI), studied in \cite{Liew22_TIT,Liew24_TIT} for scheduled monitoring of Markov processes, and defined as the entropy at the monitor on the current state of a source conditioned on the last received sample. The metric is akin to what considered here in one of the cases we target, i.e., when the receiver only has access to the last received value and to the time elapsed since its retrieval. 
The novelty of our work is to consider imperfect observations, whose accuracy depends on spatial aspects, i.e., the distance from the event of interest, and the interactions between accuracy and success rate over a collision channel.
To properly characterize these, we consider the case in which the whole history of observations is available at the receiver up to the present time, and explore how this can significantly improve performance.

\subsection{Contribution}
The goal of this work is to move beyond purely temporal notions of freshness, as captured by metrics such as \ac{AoI} or \ac{VoI}, by explicitly incorporating a spatial dimension. These two aspects jointly influence the uncertainty at the receiver, which we quantify through the conditional entropy of the system state, given the history of received messages.

This setting is particularly relevant in \ac{IoT} data collection scenarios, where a gateway asynchronously gathers updates from spatially distributed sensors. We consider a model in which sensors monitor a Markov process and access the channel via slotted ALOHA, providing a tractable yet representative framework for more complex systems. We leverage a general distance-based error model for sensor observations, and formulate an analytically tractable framework that  preserves the spatio-temporal characteristics of the scenario, which can be characterized through a \ac{HMM}. The resulting model has a single-sensor equivalent, whose activity rate and accuracy reflect the spatial distribution of nodes and their activation probabilities. We thus optimize the medium access parameters, analyzing the effect of the coverage radius and characterizing the optimal access probability for different regions as well as the effect of including geographical knowledge into transmitted packets.
Furthermore, we analyze a \emph{forgetful} receiver that only keeps the last received packet and a \emph{full-sequence based (FS)} receiver that considers the full history of received data up to the present moment.\footnote{These two receivers are often colloquially referred to as \emph{stateless} and \emph{stateful}, respectively. We avoid this terminology in the present work, as the belief of the forgetful receiver also evolves over time following the Markov chain that represents the process, and the use of ``stateless'' would thus be formally incorrect.}

To the best of our knowledge, our work is the first to fully characterize this type of scenario in terms of receiver uncertainty, providing a useful tool for system optimization. 
In particular, we show that restricting the coverage range to a smaller value may reduce said uncertainty at the expense of system throughput. For most of our work, we concentrate on binary Markov processes. Aside from being relevant in and of itself for applications such as anomaly detection, the binary case allows us to analyze and understand trends and obtain insights that can carry over to more complex models. We then extend our results to a general Markov model with $M>2$ states. Our main contributions can then be summarized as follows:
\begin{itemize}
    \item We observe that medium access is memoryless across slots, which can be captured by two parameters, the success probability and accuracy of updates. The problem could also be modeled as a single sensor with those aggregate properties. We characterize the spatio-temporal observation process and provide the parameters corresponding to different medium access settings;
    \item We analyze the uncertainty of a forgetful and full-sequence based receiver for a binary Markov process by computing the conditional entropy of the source state, given the observations available at the receiver. The forgetful approach can be of practical relevance in view of its simplicity and allows us to obtain an exact, closed-form expression of the metrics of interest. For the full-sequence based case, uncertainty is captured by formalizing an \ac{HMM} describing the observations collected by the receiver;
    \item For both solutions, we study in detail the effect of the asymmetry of the binary Markov process and prove the existence of an optimal coverage range to minimize the average uncertainty at the receiver. Notably, the solution differs from the system configuration that would be attained by minimizing \ac{AoI}, which also maximizes the system throughput, clarifying the critical role of considering the spatial dimension in information freshness;
    \item The initial model is extended to capture the case in which devices are aware of their location, and make this information available whenever sending an update. The receiver can then leverage such geographical knowledge to weigh the reliability of the incoming readings. Moreover, we also determine how the transmission probabilities should be adapted based on the nodes' distance from the sensed process, in order to minimize the receiver uncertainty. Our study provides non-trivial insights, revealing that both cases lead to negligible performance gains, hinting at the design of simpler \ac{IoT} systems;
    \item We finally extend our main results to a Markov process with $M$ states, showing that the main considerations we draw are not due to the simplistic nature of the binary process, but represent more general properties of spatio-temporal monitoring systems.
\end{itemize}
A preliminary version of this work was presented in~\cite{Munari25_SPAWC}, which contained only the characterization of the forgetful receiver.
The rest of this paper is organized as follows: first, Sec.~\ref{sec:sysModel} presents the system model and metrics. The analytical results for the forgetful and full-sequence receiver are then derived in Sec.~\ref{sec:forgetful_analysis} and Sec.~\ref{sec:mindful_analysis}, respectively.  Sec.~\ref{sec:optim} presents the optimization of the communication systems, and Sec.~\ref{sec:multiState} describes the results when extending the monitoring to an $M$-state process. Finally, Sec.~\ref{sec:concs} concludes the paper.

\section{System Model and Preliminaries}
\label{sec:sysModel}

\subsection{Notation}
We denote a discrete random variable (r.v.) and its realization by upper- and lowercase letters, respectively, e.g., $X$ and $x$. The probability mass function of the r.v. $X$ is indicated as $\pr[X=x] = p(x)$, and the conditional distribution of $X$ given $Y$ is $p_{X\givenS Y}(x \given y)$. Subscripts are omitted when no ambiguity arises. The Shannon entropy of a r.v. $X$ is denoted by $H(X)=-\sum_x p(x) \log_2 p(x)$, with the usual extensions to conditional entropy. When $n$ denotes the current time index, $X_n$ is the corresponding value of a random process, whereas we indicate by $X^n = [X_0, X_1, \dots, X_n]$ the vector of all the values of the process up to $n$. 
Finally, we denote vectors in lowercase boldface, e.g., $\bm a$. In particular, a vector denoting the position of a point in $\mathbb R^2$ is indicated as $\bm x$.

\subsection{System model}\label{ssec:model}

Consider a physical process of interest, modeled as a finite-state, time-homogeneous discrete time Markov chain \Mcn, $n\in \mathbb N$, taking values in \McA. We denote the one-step transition probability matrix for \Mcn\ as \TransMat. For the initial part of our study, we focus on the case of a two-state chain with $\mathcal{X}=\{0,1\}$, so that the transition matrix is
\begin{align}
    \TransMat =
    \begin{pmatrix}
        1-q         & q \\
        \asymm q    & 1-\asymm q
    \end{pmatrix}
    \label{eq:transMat}
\end{align}
with $q>0$, leading to the stationary distribution of the ergodic chain \mbox{$\pi_0 = \asymm/(1+\asymm)$}, \mbox{$\pi_1 = 1-\pi_0$}.  In \eqref{eq:transMat}, $\asymm \in (0,1/q)$ is the \emph{asymmetry factor}, denoting the ratio of the time spent in state $0$ to the time spent in state $1$. Without loss of generality, we will consider in the remainder $\eta\geq 1$. The case $\eta=1$ identifies symmetric source, whereas, for $\eta>1$, state $0$ is more frequently visited. The case $\eta<1$ means that state $1$ is more frequent than $0$, so we can simply swap the indices of the states. The general setup with $\lvert \McA \rvert > 2$ will be tackled in Sec.~\ref{sec:multiState}, introducing the required modeling details. 

We study a network topology composed of a receiver and a population of $\nodes\in\mathbb{N}$ \ac{IoT} devices, which are spread over a circular coverage area \area\ of radius \radMax, centered without loss of generality in the origin of the plane. Nodes are then randomly distributed over the coverage area according to the uniform distribution \ac{PDF}
\begin{equation}\label{eq:uniform_dist}
    f(x,y)=\begin{cases}
        \pi^{-1}\radMax^{-2}, &\text{if } (x,y)\in\area;\\
        0, &\text{otherwise.}
    \end{cases}
\end{equation}
At any time slot, each device within coverage independently decides whether to take a reading of the process and transmit a message containing the value to the receiver.\footnote{Nodes outside the coverage area are inactive and do not attempt any communication with the receiver.}  

Due to the displacement of devices within \area, the reliability of the observations produced by a node depends on its position. In particular, we assume that devices closer to the origin of the plane provide more precise readings, whereas the accuracy of messages coming from farther away progressively reduces, as the longer the information about the event propagates physically, the higher the probability that it may be distorted. Such a setting captures, for instance, an application tracking the value of a physical quantity at the specific coordinates of interest (plane origin), or a monitoring system that has to detect an event (e.g., presence/absence of an object) in the corresponding surroundings. 

We partition the coverage area \area\ in $K$ regions: 
\begin{align}\label{eq:region_def}
    \area_i := \{ \bm x \in \area: i \rad \leq \norm{\bm x}{2} < (i+1)\rad  \,\}
\end{align}
for $i\in \mathcal K = \{0,\dots,K-1\}$, and where $\rad:=\radMax/K$. Region $\area_i$ is thus an annulus with internal radius $i\rad$, external radius $(i+1)\rad$, and area $A_i=(2i+1)\pi\rad^2$. The overall area of the coverage region is $A=\pi\radMax^2=\pi K^2\rad^2$. We can thus consider the number of nodes in area $\area_i$, denoted as $\nodes_i$, as the outcome of an urn model \cite{johnson1977urn}. Specifically, the portion of the $\nodes$ total sensors falling into that specific annulus follows a binomial distribution, i.e., $\nodes_i\sim\text{Bin}\left(\nodes;\frac{A_i}{A}\right)=\text{Bin}\left(\nodes;\frac{2i+1}{K^2}\right)$, as the probability of any given node falling into $\area_i$ is $A_i/A$, due to the uniform distribution across the coverage area following~\eqref{eq:uniform_dist}. The expected number of nodes in $\area_i$ is thus $\mathbb{E}[\nodes_i]=\frac{(2i+1)\nodes}{K^2}$.

Whenever a device in $\area_i$ performs a reading of the process \Mcn, it obtains a value \Readn\ according to the conditional distribution $p(\readn\given \mcn, i)$. In the case of the two-state process under study, the PMF simplifies to having a correct ($\readn=\mcn$) or wrong ($\readn\neq\mcn$) reading. For the time being, we assume that the reliability of \Readn\ does not depend on the current value of the source, and we model
\begin{align}
    \lambda_i := \pr\left[\Readn = \Mcn \given i\right] = \max \left\{ \frac{1}{(1+i\rad)^{\powLaw}}, \frac{1}{2} \right\}
    \label{eq:powerLaw}
\end{align}
with $\alpha > 0$, so that $\pr \left[\Readn \neq \Mcn \given i\right] = 1-\lambda_i$. The extension to the case in which the accuracy of a sensor is a continuous, monotonically decreasing function of the distance is relatively straightforward, replacing \acp{PMF} with \acp{PDF} and sums with integrals in the following.
This power-law characterization has been shown to accurately capture spatial correlation among readings provided by sensor nodes for physical processes of interest \cite{Hribar2018_IoT}, although the analysis presented can be adapted to different reliability functions. 

The nodes send updates to the receiver over a shared wireless channel, following a slotted ALOHA policy \cite{Abramson77:PacketBroadcasting}. Accordingly, at the beginning of a slot, each node independently decides with probability \pTx$\in(0,1)$\ to sense the process and transmit a packet containing the obtained reading. As the transmission decision does not depend on the state of the underlying process or on its observation, the transmission process is memoryless. Following the on-off fading channel model \cite{OnOff2003,Sun16:PECCSA,Ivanonv17:TCOM}, a packet is erased with probability $\peras<1$, bringing no power contribution to the receiver, or arrives unfaded with probability $1-\peras$. We further regard collisions to be destructive, i.e., a slot in which two or more packets arrive non-erased at the receiver prevents decoding of any of them, while singleton packets arriving unfaded over a slot are successfully received~\cite{Abramson77:PacketBroadcasting}.
For the considered model, a slot sees the delivery of a message as an i.i.d. Bernoulli random variable with success probability
\begin{align}
    \ps = \nodes\; \pTx \; (1-\peras) \cdot [1-\pTx(1-\peras)]^{\nodes-1},
    \label{eq:psucc}
\end{align}
corresponding to the average throughput of the system.
At the receiver side, no knowledge can be gathered on the area out of which a received message was generated, e.g., due to location-unaware devices, or to the fact that no geographical information is piggybacked in transmitted packets. This assumption will be relaxed in \secr\ref{sec:locationAware}.

\begin{remark}
    The considered model captures setups in which packets coming from different nodes are equally likely to be retrieved. For instance, this is the case in settings where the relative displacement of nodes is not predominant on the overall path-loss. This includes ground nodes communicating with a satellite~\cite{lee2024handover} or long-range IoT scenarios (e.g., LoRa~\cite{callebaut2019characterization}) in which the event of interest is not co-located with the base station. The model has been widely used in the literature~\cite{beltramelli2020lora,yavascan2021analysis,testi2025packet} to abstract the specific type of physical layer being considered, aiming at results that can be applied to a wide class of systems, and has been proven to correctly identify the fundamental trade-offs that arise in the presence of random access contention.
\end{remark}

Interestingly, one could construct a single-sensor model whose performance is equivalent to the full system: this single-sensor equivalent summarizes the spatial aspect into the accuracy of individual observations, which results from the measurement distortion function described above, and the temporal aspect in the throughput \ps. 
Once the decentralized transmission strategy is established, this equivalent model allows for an abstract performance analysis based on parameters that depend on spatial patterns and medium access design.

\begin{remark}
The existence of a single-sensor equivalent depends only on the memoryless nature of the transmission process, and would also hold in the presence of distance-dependent reception or more complex channel models. Since the focus of the present work is on the spatial accuracy of the readings, the model we target allows us to factorize the effect of \ps\ and of the accuracy of the received measurements. In the general case, the joint distribution of accuracy and reception probability would need to be considered.
\end{remark}

In the remainder of the analysis, we will often focus on a generic slot $n$ over which the receiver decodes an incoming message (singleton). Denote this event as $\mathcal R_n$ and note that it occurs with probability $\ps$. For the sake of compactness, and to simplify notation, we will indicate the distribution of a r.v. $A_n$ conditioned on $\mathcal R_n$ as $\pb(a_n) = \pr[ A_n = a_n \given \mathcal R_n ]$.

In this context, our aim is to characterize the uncertainty the receiver has about the current state of the source. The metric is relevant for a number of IoT applications, e.g., dealing with actuation and decision-making, and, as will be illustrated later, is an interesting proxy to capture both the temporal and spatial components of information freshness for the system model under study. In the remainder, we will tackle two possible implementations at the receiver, dubbed \emph{forgetful} and \emph{full-sequence based} (FS). In both cases, we will capture the uncertainty considering information available at the application layer and assume that only the content of successfully decoded messages is forwarded from the lower layers of the receiver. Moreover, we assume that the number of users in the coverage area, the channel access probability, and the source statistics, as well as the reliability law in \eqref{eq:powerLaw}, are known at the receiver.\footnote{Assuming knowledge of \nodes, \peras\ and \pTx\ is reasonable in practical systems. The possibility of remote source monitoring without prior knowledge of source statistics has also been recently studied, e.g., in \cite{Asgari25}.}

\begin{remark}
    The modeling approach presented is inspired by IoT applications, in which control decisions are taken at the application layer, and little to no cross-layer information is exchanged, e.g., presence of idle vs. collided slots. To leverage this additional knowledge, which could help reduce receiver uncertainty by defining transmission policies linked to the source evolution \cite{Cocco23_JSAIT,Munari25_ISIT}, specific cross-layer interfaces must be implemented. As we target practical system implementations, we do not tackle this aspect in the present work.
\end{remark}

\subsection{Receiver entropy}
Let $\Rcn$ denote the channel output that the receiver observes at time $n$. 
The r.v. has alphabet $\mathcal X \cup \{\bigast\}$, as
\Rcn\ belongs to the Markov process state space $\mathcal X$ when the receiver decodes a packet containing a reading of the tracked process, whereas symbol $\bigast$ is used to represent the output of a slot in which no packet is decoded.
In this case, the uncertainty at time $n$ given the current realizations of the involved processes is 
\begin{align}
    \condent(\rcnvec) := H(\Mcn \given \Rcnvec = \rcnvec),
    \label{eq:cond_ent_def}
\end{align}
where we recall that $Y^n$ is the sequence of all past observations up to and including slot $n$. 

To capture average system performance, we consider the conditional entropy on the source state given the random vector $Y^n$, $H(X_n\given Y^n)$. Specifically, we are interested in its limiting behavior as $n$ grows, denoted as
\begin{align}
    \entFun := \lim_{n\to\infty} H(X_n\given Y^n).
    \label{eq:ent_mindful_def}
\end{align}
The existence of the limit is proven in Appendix \ref{app:lemma1}.

\begin{remark}
    The considered metric characterizes the uncertainty on the true state of the process conditioned on the available observations. As such, it is not directly tied to any specific estimator, and shall not be confused with the uncertainty on the estimate $\hat X_n$ the receiver may extract (i.e, with $H(\hat X_n)$). On the other hand, $H(X_n\given Y^n)$ is by definition driven by the conditional distribution $p(x_n\given y^n)$, which can be used to implement a maximum a posteriori (MAP) estimator, minimizing the probability of error. The value of $\condent(y^n)$ can then be used to gauge the level of uncertainty at the receiver, e.g., when deciding whether to take an action as response to the source state based on the MAP estimate. 
    It shall also be noted that the metric provides broader insights on the problem of remote monitoring and estimation, giving bounds on the average error probability $P_e$. Specifically, for the binary source $X_n$ under study, we have
    \begin{align}
       H_b^{-1}( H(X_n\given Y^n) ) \stackrel{(a)}{\leq} P_e \stackrel{(b)}{\leq} \frac{1}{2} H(X_n\given Y^n),
       \label{eq:bounds}
    \end{align}
    where $H_b(p) = -p\log p -(1-p)\log(1-p)$ is the binary entropy function. Within \eqref{eq:bounds}, (a) follows from Fano's inequality \cite{Cover_Thomas}, characterizing a lower bound on $P_e$ for any estimator. Even more relevantly, Kovalevsky's result in (b) bounds from above the average error probability of a MAP estimator \cite{kovalevsky,Feder94_TIT}.\footnote{The bound is valid for $P_e\leq 1/2$, which is always true for the considered binary $X_n$.} 
    The inequalities in \eqref{eq:bounds} corroborate the use of the conditional entropy $H(X_n\given Y^n)$ as metric for protocol design in remote monitoring applications, as its reduction directly improves the attainable performance in terms of error probability.
\end{remark}

\subsection{Age of Information (AoI)}
\label{sec:aoi}
The metric defined in \eqref{eq:ent_mindful_def} aims to capture how both the spatial and temporal freshness components of the information available to the receiver influence its uncertainty. In the remainder, it will also be useful to consider a measure of the sole temporal dimension. This is provided by the instantaneous \ac{AoI} $\Delta_n$, defined as the number of slots since the reception of the latest packet, and by its average $\avgAge := \mathbb E[\Agen]$. Recall that in the system model considered: \textit{(1)} a message is generated at the time of transmission; and, \textit{(2)} the slotted ALOHA channel access probability is the same for all nodes and independent of the collected reading. The average AoI can be characterized in terms of the success probability in~\eqref{eq:psucc} as~\cite{yates2017status}
\begin{align}
    \avgAge = \frac{1}{2} + \frac{1}{\ps}.
    \label{eq:age}
\end{align}

\section{Analytical Characterization of the Receiver Uncertainty: Forgetful Case}

\label{sec:forgetful_analysis}

The first receiver type we consider monitors the source solely based on the last received message, i.e., it is oblivious of previously received information. Following this approach, the receiver in time slot $n$ has knowledge about: \textit{(1)}, the time elapsed since the last reception; and \textit{(2)}, the content of the last decoded packet.
Recalling that we assume that transmitting nodes sample the process at the beginning of the slot used to attempt delivery, \textit{(1)} also corresponds to the current \ac{AoI} at the receiver \Agen. Moreover, we indicate the reading contained in the last decoded message as of time $n$ as $\lastRcn \in \mathcal X$. In other words, if a new reading is received over a slot, the quantity is reset to the obtained update, whereas it retains the value of the last retrieved message if no packet is decoded over the current time unit. The uncertainty of the \emph{forgetful} receiver in slot $n$, given that a specific realization is observed, is then captured by the entropy
\begin{align}
    \condentforg(\lastrcn,\agen) := H(\Mcn \given \lastRcn = \lastrcn, \Agen=\agen).
    \label{eq:condent_forgetful}
\end{align}
The quantity $\condent_f$ is relevant from an operational point of view, as it can serve as a basis for properly gauging current knowledge or making actuation decisions. Note indeed that, in the forgetful case, the posterior distribution $p(x_n\given w_n,\delta_n)$ suffices to derive a MAP estimator, and the entropy metric can provide a measure of the uncertainty on the true source state to be used in combination with the available estimate.
We also note that $\condent_f$ reduces to the uncertainty of information introduced in~\cite{Liew22_TIT,Liew24_TIT} in the case of perfectly reliable readings. 
In this setting, we are interested in the metric $\entFun_f$, capturing the limiting behavior of the conditional entropy $H(X_n\given W_n,\Agen)$. We characterize it with the following result.
\begin{lemma}
    For the forgetful receiver, we have
    \begin{align}
        \entFun_f := \lim_{n\to\infty} H(X_n\given W_n,\Agen) = \sum_{w_n,\agen} \condent_f(w_n,\agen) p(w_n,\agen),
            \label{eq:ent_forgetful}
    \end{align}
    where $p(w_n,\agen)$ is the joint stationary distribution of the last received symbol $W_n$ and the AoI $\Agen$.
\end{lemma}
\begin{IEEEproof}
   See Appendix \ref{app:lemma2}.
\end{IEEEproof}

While clearly suboptimal, this type of receiver is of interest for a number of reasons. Firstly, it allows for the derivation of compact closed-form expressions for the metrics of interest. Moreover, its discussion will be instrumental to introduce the fundamental trade-offs that characterize the system, such as the effect of tuning the coverage radius by preventing nodes farther than a certain distance from transmitting, which reduces the system throughput, but may increase the receiver's accuracy. Secondly, the solution epitomizes the behavior of practical \ac{IoT} systems, in which low-complexity approaches are often preferred. In this perspective, the values of $\condentforg(\lastrcn,\agen)$ could easily be pre-computed and stored in a look-up table, allowing the receiver to gauge the current level of the entropy in \eqref{eq:condent_forgetful} without any computational effort. Finally, the forgetful receiver provides a meaningful benchmark when readings are reliable, as clarified by the following result.
\begin{remark}\label{rem:ideal} 
Following the data processing inequality, we have $\condentforg(\lastrcn,\agen)\geq \condent(\rcnvec)$, as $(\lastrcn,\agen)$ can be deterministically derived from $\rcnvec$. In this case, packet receptions represent renewals from the receiver's perspective, as all information about previous steps is forgotten. If sensor readings are always reliable, the forgetful receiver minimizes the entropy, in the sense that $ \condentforg(\lastrcn,\agen)=\condent(\rcnvec)$, as the next steps of the process depend only on the last observed state thanks to the Markov property.
\end{remark}

\begin{example}
\label{ex:forgetful}
    Consider the tracking of a symmetric source ($\eta=1$), with transition probability \mbox{$q=0.05$}. Two examples of how $\condentforg(\lastrcn,\agen) = H(\Mcn\given \lastRcn=\lastrcn,\Agen=\agen)$ may evolve over time are reported in \figr\ref{fig:timeline}, assuming a reliability power-law coefficient $\alpha=0.02$, $\peras=0.1$, and setting a transmission probability $\pTx=10^{-4}$. In the first case, \figr\ref{subfig:R6}, a smaller radius of $\radMax=6\rad$ was considered, whereas the rightmost plot, \figr\ref{subfig:R15}, was obtained for a wider coverage of $\radMax=15\rad$, with $\rad=10$~m. The total number of nodes was scaled proportionally to the area, maintaining the same average number of nodes per unit area, which we denote as $\dens=\frac{\nodes}{\pi\radMax^2}$. In both examples, $\dens\simeq0.05$, up to a negligible error due to the integer number of nodes. The example hints at a key trade-off. On the one hand, selecting a smaller radius leads to more sporadic update deliveries (corresponding to resets of $\condentforg(\lastrcn,\agen)$), and the uncertainty tends to grow for longer periods of time. 
    Whenever the receiver does not retrieve messages for some time, such as in the central part of \figr\ref{fig:timeline}(a), its uncertainty tends to converge to the stationary entropy of the source, i.e., $ \ent(X) = -\pi_0 \log_2 \pi_0 - \pi_1 \log_2 \pi_1 = 1 \text{ [bit]}$. This effect is much more subdued in \figr\ref{subfig:R15}, as the presence of more nodes in the system results in an improvement of the frequency of message delivery, with a lower \ac{AoI} and an uncertainty that is more often brought back to its minimum. On the other hand, such reset value gets higher as \radMax\ increases. In fact, incoming messages are more likely to be generated by devices farther away from the origin, containing less reliable readings as per \eqref{eq:powerLaw}, leaving the receiver with a higher residual uncertainty. The critical balance of these two factors will determine system performance, epitomizing the role of spatial and temporal information freshness.
\end{example}

\begin{figure}
    \centering
    \subfloat[Coverage range $\radMax = 6\rad$, $\nodes=565$.\label{subfig:R6}]{
        \includegraphics[width=.26\columnwidth]{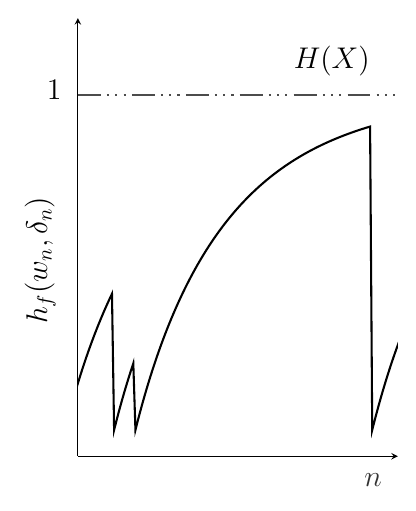}}
    \hspace*{3em}
    \subfloat[Coverage range $\radMax{=}15\rad$, $\nodes=3534$.\label{subfig:R15}]{
        \includegraphics[width=.26\columnwidth]{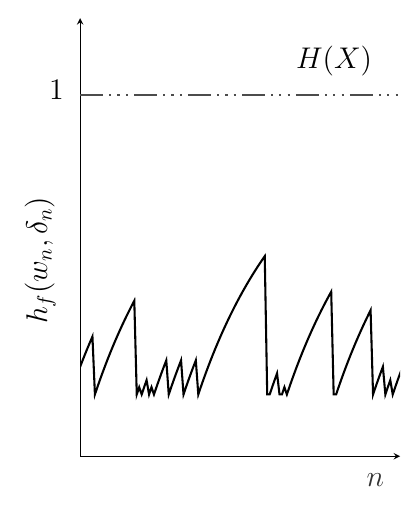}}
    \caption{Examples of time evolution of $\condent_f(\lastrcn,\agen)$ when tracking a symmetric source. The horizontal line reports the value of the stationary source entropy, $H(X) = -\pi_0\log_2 \pi_0 - \pi_1 \log_2 \pi_1$.
    The plot was generated for $\alpha=0.02$, $\rad=10$m, $q=0.005$, $\dens\simeq0.05$, $\pTx=10^{-4}$, $\peras=0.1$.}
    \label{fig:timeline}
    \vspace{-1em}
\end{figure}

\subsection{Derivation of the average uncertainty, $\entFun_f$}

To characterize the receiver uncertainty, we focus first on a generic time $n_0$ over which a message is successfully decoded. Accordingly, the slot sees the reset of: \textit{(1)}, the AoI value to $\Delta_{n_0}=0$, and \textit{(2)}, $\lastRc_{n_0}$ to the value contained in the incoming reading. Applying Bayes' rule, the distribution of the state of the tracked process conditioned on the received reading can be expressed as 
\begin{align}
    \pb(\mc_{n_0}\given \lastrc_{n_0}) = \frac{\pb(\lastrc_{n_0}\given \mc_{n_0}) \, \pi_{\mc_{n_0}}}{\sum_{\mc_{n_0}^{\prime}\in \mathcal X} \pb(\lastrc_{n_0}\given \mc_{n_0}^{\prime}) \, \pi_{\mc_{n_0}^{\prime}}},
    \label{eq:pXnGivenYnReset}
\end{align}
where we recall that the notation $\pb(\cdot)$ was introduced in Sec. \ref{sec:sysModel}-A to denote a distribution conditioned on having decoded a packet during the current slot. Within \eqref{eq:pXnGivenYnReset}, $\pi_{\mcn}$ is the stationary distribution of \Mcn, introduced in Sec. \ref{sec:sysModel}, and we have used the fact that the event of successfully receiving a packet is independent of \Mcn\ for the transmission policies under study.
In turn, $\pb(\lastrc_{n_0}\given \mc_{n_0})$ can be conveniently expressed in terms of the r.v. $D_{n_0} \in\mathcal K$, which captures the region from which the message originated. Leaning on the law of total probability, and on the fact that transmissions are independent on the state of the source:
\begin{align}
    \pb(\lastrc_{n_0}\given \mc_{n_0}) = \sum_{\dist_{n_0}=0}^{K-1} \pb(\lastrc_{n_0}\given \mc_{n_0},\dist_{n_0}) \, \pb(\dist_{n_0}).
    \label{eq:pYnGivenXnReset}
\end{align}
Following the urn model from Sec.~\ref{ssec:model}, we have
\begin{align}
    \pb(\dist_{n_0}) =  \frac{\mathbb{E}[\nodes_i]}{\nodes}=\frac{2\dist_{n_0}+1}{K^2}.
    \label{eq:pDnGivenXn}
\end{align}
In turn, the first factor within the summation in \eqref{eq:pYnGivenXnReset} captures the probability that the decoded message contains value $\lastrc_{n_0}$, conditioned on the actual state of the process and the reliability of the sender. Accordingly, we can write in a compact form
\begin{align}
    \pb(\lastrc_{n_0}\!\given\!\mc_{n_0},\dist_{n_0}) =  \lambda_{\dist_{n_0}}\mathbbm{1}(\mc_{n_0},\!\lastrc_{n_0})
     + (1{-}\lambda_{\dist_{n_0}})\mathbbm{1}(\bar{\mc}_{n_0},\!\lastrc_{n_0}),
     \label{eq:pYnGivenDnXn}
\end{align}
where $\mathbbm{1}(a,b)$ is the indicator function, taking value $1$ if $a=b$ and $0$ otherwise. 
For the power-law in \eqref{eq:powerLaw}, combining \eqref{eq:pDnGivenXn} and \eqref{eq:pYnGivenDnXn}, we obtain
\begin{equation}
    \label{eq:pYnGivenXnReset_law}
    \pb(\lastrc_{n_0}\given \mc_{n_0}) = \sum_{\mathclap{\dist_{n_0}=0}}^{\mathclap{K-1}}\ \frac{2\dist_{n_0}+1}{K^2(1+\dist_{n_0} \, \rad)^{\alpha}}, \text{ if } \lastrc_{n_0}=\mc_{n_0}
\end{equation}
and its complementary value for the case $\lastrc_{n_0}\neq \mc_{n_0}$.
Plugging \eqref{eq:pYnGivenXnReset_law} into \eqref{eq:pXnGivenYnReset} thus provides the statistics available at the receiver on the state of the process of interest upon reception of a packet. Recalling that $\pb(\mc_{n_0}\given\lastrc_{n_0}) = p_{\Mcn\given\lastRcn,\Agen}(\mc_{n_0}\given\lastrc_{n_0},0)$, the result allows to derive the general probability $p(\mcn\given \lastrcn,\agen)$ at any slot $n=n_0+\agen$ between $n_0$ and the successive reception. This is obtained as the $\agen$-step evolution of the transition matrix $T$ with initial distribution $\pb_{\Mcn\givenS\lastRcn}(0\given\lastrc_{n_0})$ and $\pb_{\Mcn\givenS\lastRcn}(1\given\lastrc_{n_0})$.

By its definition in~\eqref{eq:condent_forgetful}, the uncertainty of the forgetful receiver on the current state of the process given the currently available knowledge can then be computed as
\begin{equation}
    \condent_f(\lastrcn,\agen) = -\!\!\sum_{\mcn\in\mathcal X} p(\mcn\given\lastrcn,\agen) \log_2 p(\mcn\given\lastrcn,\agen).
    \label{eq:condEntropy_def}
\end{equation}

Finally, we derive the conditional entropy $\entFun_f$. From \eqref{eq:ent_forgetful}, this can be done computing the PMF $p(\lastrcn,\agen)$, capturing, for a generic time slot $n$ under stationary conditions, the joint probability of having decoded the last message \agen\ slots ago and of its content being \lastrcn. Observing that packet receptions are i.i.d. across slots with probability \ps, the current \ac{AoI} value is independent of the collected reading, as each node's transmission probability does not depend on its age or position.
Consequently, $p(\lastrcn,\agen) = p(\lastrcn) p(\agen)$. Moreover, $p(\lastrcn)$ can simply be computed from \eqref{eq:pYnGivenXnReset} as
\begin{align}
    p(\lastrcn) = \sum\nolimits_{\mcn\in\mathcal X} \pb(\lastrcn\given \mcn) \,\pi_{\mcn}.
    \label{eq:pYnStat}
\end{align}
In turn, the receiver experiences an \ac{AoI} $\agen\geq 0$ at the start of slot $n$ if the last success was followed by $\agen$ failures, so that:
\begin{align}
    p(\agen) = \ps (1-\ps)^{\agen}.
    \label{eq:pAgenStat}
\end{align}
Combining \eqref{eq:pYnStat}-\eqref{eq:pAgenStat} with \eqref{eq:condEntropy_def} provides an exact formulation of $\entFun_f$.


\begin{remark} The approach followed in the analysis above considers an i.i.d. distribution for the random variable $D_n$ across slots. This implicitly corresponds to independently re-drawing the positions of all nodes at each time unit, and thus introduces an approximation of the behavior of a practical topology in which the locations of the devices are static. The validity of this assumption will be discussed in Sec.~\ref{sec:results_forgetful}.
\label{remark_iid}
\end{remark}

\subsection{Numerical results and key trends}
\label{sec:results_forgetful}

To gauge the impact of spatio-temporal information freshness on the forgetful receiver, we focus on a population of devices spread over a variable coverage area \area, scaling the number of nodes by maintaining a fixed average number of nodes per unit area $\dens\simeq 5\cdot 10^{-2}$ [nodes/m$^2$], up to a negligible approximation error given by the constraint that $\nodes$ must be an integer number, transmitting with probability $\pTx = 10^{-4}$ at every slot.\footnote{The chosen transmission probability is representative for IoT applications, corresponding to the behavior of a LoRa device sending packets of duration $300$ ms roughly once per hour.} A sent packet is erased with probability $\peras = 0.1$, and the tracked Markov process \Mcn\ transitions from state $0$ to $1$ with probability $q=5\cdot 10^{-3}$. Unless otherwise specified, we assume that the nodes produce readings following the reliability law in \eqref{eq:powerLaw}, where $\alpha = 0.02$ and $\rad=5$ m.

\begin{figure}
    \centering
    \includegraphics[width=.63\columnwidth]{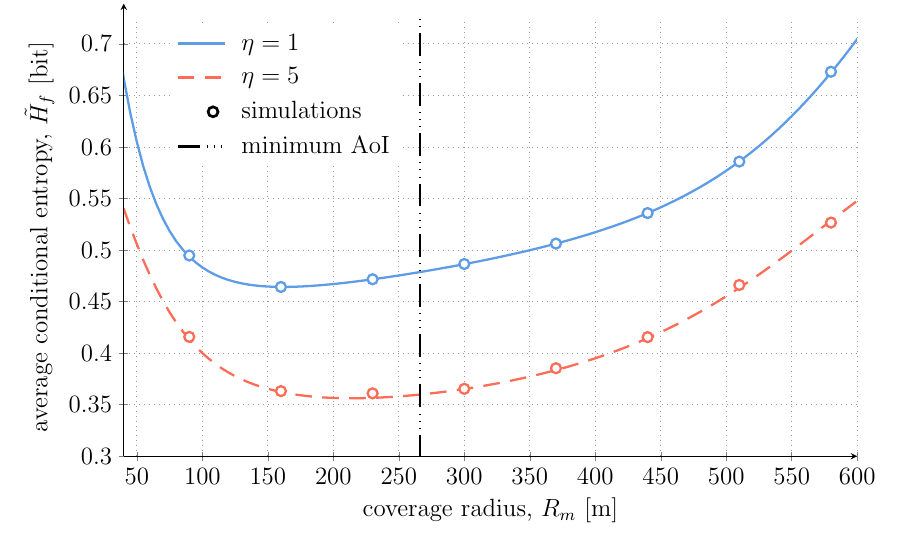}
    \caption{Conditional entropy $\entFun_f$ vs. coverage radius $\radMax$, when tracking a symmetric (dashed line) or an asymmetric (solid line) process, considering a forgetful receiver. The markers report the outcomes of simulation results, and the vertical line indicates the coverage radius that would minimize \ac{AoI}.}
    \label{fig:HvsRad}
\end{figure}

Some first fundamental insights emerge from \figr\ref{fig:HvsRad}, where we report the conditional entropy $\entFun_f$ against the coverage radius \radMax. The solid line reports the case of monitoring a symmetric process ($\eta=1$), whereas the dashed curve was obtained for an asymmetric case with $\eta = 5$. Circular markers show the results of detailed network simulations obtained at a subset of \radMax\ values, and aimed at verifying the approximation introduced in our analysis (see Remark \ref{remark_iid}). Specifically, for each simulation, the positions of the $\nodes$ nodes were drawn uniformly within \area, and their location was kept fixed for the whole duration of the run ($5\cdot10^4$ slots). Throughout the run, the Markov process evolved according to the transition matrix \TransMat, and the time evolution of $\condent_f(\lastrcn,\agen)$ was computed via \eqref{eq:condEntropy_def}. The value of $\entFun_f$ was then obtained as the time average of $\condent_f$ as per \eqref{eq:condent_forgetful}. Each of the reported marked points shows in turn the average of $50$ independent topologies. The excellent match between analysis and simulations under all conditions confirms the tightness of the modeling assumptions.

More interestingly, \figr\ref{fig:HvsRad} reveals the existence of an optimal coverage range to minimize $\entFun_f$. When $\radMax$ is too small, the updates are sporadic as few nodes fall within \area. In this case, even if incoming messages are highly reliable in view of the proximity of active devices to the origin, the receiver's uncertainty is dominated by the lack of timely (fresh) updates. However, if the coverage radius is increased beyond a critical point $\radMax^*$, the conditional entropy grows once more (rightmost region of the plot). In such conditions, even if messages are received more often, their reliability decreases sharply due to the higher probability of obtaining readings from nodes farther away from the origin, leading to a greater uncertainty. Striking a proper balance between the temporal and spatial components of information freshness is thus paramount and shall represent a key criterion in the design of the system.\footnote{The coverage radius could, for instance, be controlled in LEO IoT systems by tuning the receiver beam-width via beamforming, or in terrestrial systems by admitting at logon only users that are received with a high enough (average) power level.} From this standpoint, it is relevant to note that the optimal radius $\radMax^*$ can be significantly smaller than the one that would be obtained considering \ac{AoI} only. This is reported for reference by the vertical dash-dotted line in the figure. From \eqref{eq:age}, such a configuration is clearly obtained by maximizing the system throughput \ps\ \cite{Yates17:AoI_SA,Munari21_TCOM}. This is achieved when the number of users in coverage satisfies $\nodes\pTx(1-\peras)=1$, that is, for an incoming channel load of $1$ [pkt/slot], corresponding to an optimal radius $(\pi\! \dens \pTx (1-\peras))^{-1/2}$.

\figr\ref{fig:HvsRad} also highlights a difference in optimal coverage when tracking symmetric and asymmetric sources, with the latter case presenting a lower uncertainty and a larger $\radMax^*$. This is due to the change in $p(x_n|w_n,\delta_n)$ in slots without an update: we know that $p(x_{n+1}|w_n,\delta_n+1)=\sum_{x_n\in\mc{X}}\TransMat_{x_n,x_{n+1}}p(x_n|w_n,\delta_n)$. Consequently, highly asymmetric sources tend to quickly return to the most frequent state, thus reducing the entropy in the absence of updates. On the other hand, the uncertainty over the value of the state after several steps from the last update is very high if $\eta=1$. This effect reduces the impact of the temporal component of information freshness, represented by the lower \ac{AoI}, as the receiver can tolerate higher values of the age without a significant increase in the entropy. The relative importance of obtaining accurate updates, which depends on the location that successful packets come from and represents the spatial component of our problem, thus increases, decreasing the optimal value of \radMax.

\begin{figure*}
    \centering
    \subfloat[Optimal coverage range.\label{subfig:opt_radius}]{
        \includegraphics[width=.47\columnwidth]{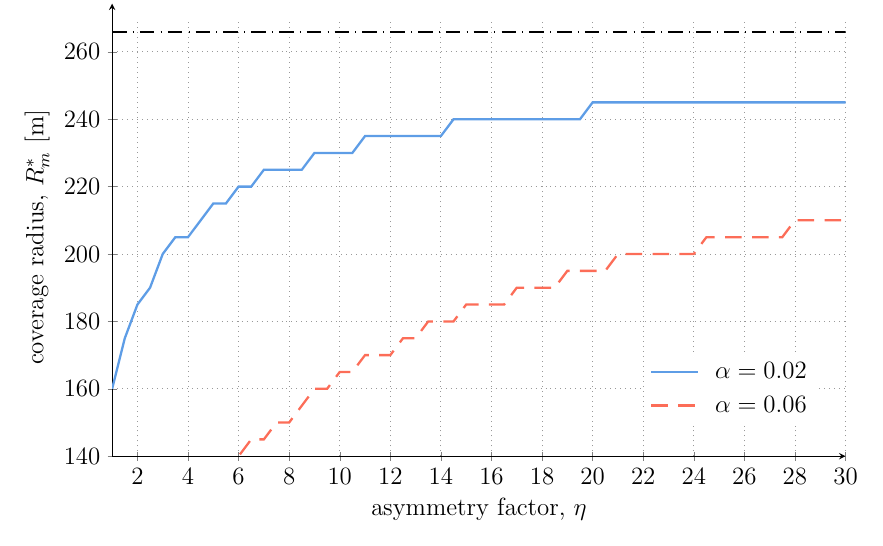}}
    \hspace*{2em}
    \subfloat[Minimum conditional entropy.\label{subfig:min_Hf}]{
        \includegraphics[width=.47\columnwidth]{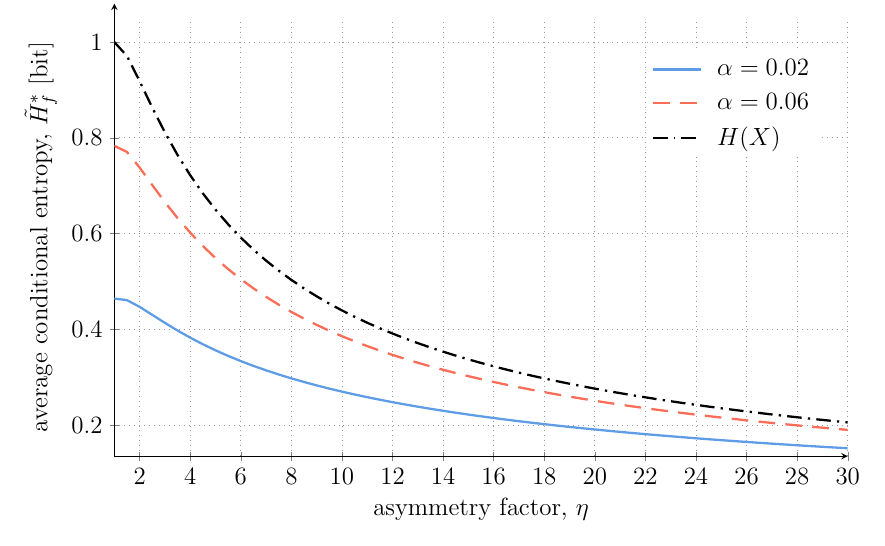}}
    \caption{Optimal coverage range $\radMax^*$ and minimum average conditional entropy $\tilde H_f^*$ obtained with $\radMax^*$ vs asymmetry factor $\eta$ for a forgetful receiver. The horizontal line corresponds to coverage for minimum AoI in (a) and to the source stationary entropy in (b).}
    \label{fig:optHandRadius}
\end{figure*}

To further inspect this phenomenon, we consider the optimal coverage radius $\radMax^*$, shown in \figr\ref{subfig:opt_radius}, and the corresponding minimum conditional entropy $\entFun_f^*$, depicted in \figr\ref{subfig:min_Hf}, against the asymmetry factor. The optimal radius increases sharply as the asymmetry grows and, eventually, $\radMax^*$ tends to converge to the value minimizing \ac{AoI}, i.e., $(\pi\! \dens \pTx (1-\peras))^{-1/2}$, since the temporal component of freshness dominates.\footnote{The stepwise behavior of the radius is due to the fact that only integer and even values were considering due to computational constraints.} A corresponding decrease in the attainable average uncertainty is experienced as $\eta$ grows, obtained by admitting more nodes in the system. \figr\ref{fig:optHandRadius} reports the trends for two different values of the exponent $\alpha$, determining the reliability of the messages incoming according to \eqref{eq:powerLaw}. It is interesting to note that for larger values of $\alpha$ the optimal coverage shrinks, as an effect of the faster degradation with distance of the quality of the updates. In this setting, trading off messages' timeliness for an increased reliability pays off. For instance, when $\eta=5$ (the same setup considered in \figr\ref{fig:HvsRad}), the optimal coverage radius reduces to less than one half of what would be optimal in terms of \ac{AoI} when $\alpha=0.06$. \figr\ref{subfig:min_Hf} also reports the stationary entropy, i.e., the entropy of the steady-state distribution of the process. As $\eta$ grows, the three values tend to converge, but the gains with $\alpha=0.02$ are quite impressive, further showing the viability of a forgetful approach in systems with significant hardware constraints.

\section{Full-Sequence Based Receiver}\label{sec:mindful_analysis}

Consider now an FS receiver monitoring the state of the source relying on the whole history of past observations. In this case, the uncertainty is driven by the conditional distribution $p(\mcn\given y^n)$. To track this, observe that, for the transmission policy under study, $p(y^n\given x^n,y^{n-1}) = p(\rcn\given\mcn)$, so that $(\Mcn,\Rcn)$ is a \ac{HMM}. The corresponding observation probabilities can be expressed as
\begin{align}
    p(\rcn\given\mcn) = 
    \begin{cases}
        \ps \, \pb_{\lastRcn\given\Mcn}(\rcn\given\mcn),        & \rcn \in\mathcal X;\\
        1 - \ps,                                         & \rcn = \bigast.
    \end{cases}
\label{eq:emission_prob}
\end{align}
Specifically, when a message is decoded over the current slot, i.e., $\rcn\in\mathcal X$, the r.v. \lastRcn\ describing the last received value is reset to \rcn, and the AoI is reset to $0$. Thus, $p(\rcn\given\mcn)$ is captured by $\ps\,\pb_{\lastRcn\given\Mcn}$, where the latter PMF was obtained in \eqref{eq:pYnGivenXnReset_law}. On the other hand, an observation $\bigast$ is obtained whenever no message is retrieved due to either a collision or lack of transmissions, i.e., with overall probability $1-\ps$.

Leaning on this, the joint distribution of the current source value and the overall vector of collected observations can be derived via the forward recursion of the \ac{HMM} \cite{rabiner1989tutorial} as
\begin{align}
    p(\mcn,y^n) = \sum_{\mathclap{x_{n-1}\in \mathcal X}} p(x_{n-1},y^{n-1}) p(\mcn\given x_{n-1}) p(\rcn\given\mcn),\ n>1
\end{align}
with initial conditions at $n=1$ set considering the stationary distribution $\pi_x$ of the tracked source, i.e.,
\begin{align}
    p(x_1,y_1) = \sum_{x_1\in\mathcal X} \pi_{x_1}\, p(y_1\given x_1).
\end{align}
The conditional distribution of the source is then obtained as 
\begin{align}
    p(\mcn \given y^n) = \frac{p(\mcn,y^n)}{\sum_{\mcn\in\mathcal X} p(\mcn,y^n)},
    \label{eq:posterior_FS}
\end{align}
allowing us to derive the receiver uncertainty $\condent(y^n)$ via \eqref{eq:cond_ent_def}.

Before engaging in the discussion of the performance of the FS receiver, we report in the following two examples, which are instrumental to understanding the behavior of the two receivers for both symmetric and asymmetric sources. 

\begin{figure*}
    \centering
    \subfloat[FS receiver, $\eta{=1}$.\label{subfig:ex_1}]{\includegraphics[width=0.225\linewidth]{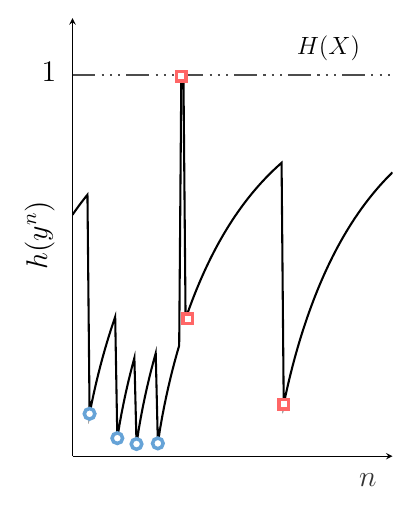}}
    \hspace{.3em}
    \subfloat[FS receiver, $\eta{=5}$.\label{subfig:ex_2}]{
    \includegraphics[width=0.225\linewidth]{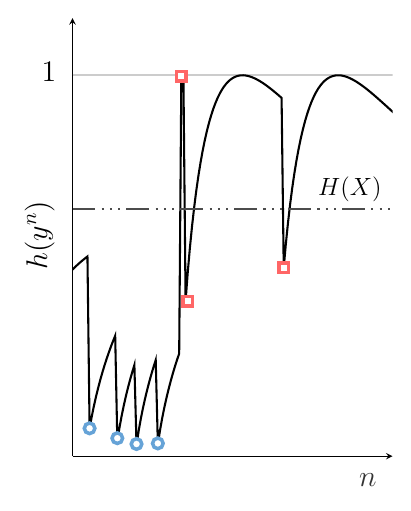}}
    \hspace{.3em}
    \subfloat[FS receiver, $\eta{=20}$.\label{subfig:ex_3}]{
    \includegraphics[width=0.225\linewidth]{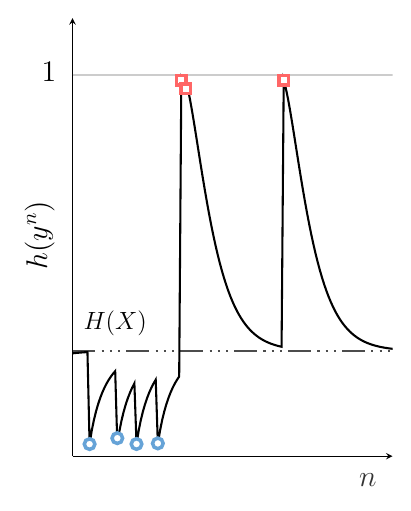}}
    \hspace{.3em}
    \subfloat[forgetful receiver, $\eta{=}5$.\label{subfig:ex_4}]{
     \includegraphics[width=0.225\linewidth]{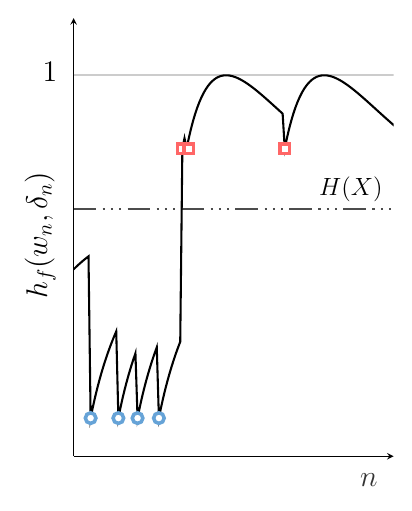}}
    \caption{Instances of time evolution of $\condent(y^n)=H(X_n|Y^n=y^n)$ and $\condent_f(\lastrcn,\agen)=H(X_n|W_n=w_n,\Delta_n=\delta_n)$, discussed in Examples \ref{ex:mindful_symm} and \ref{ex:mindful_asymm}. In all cases, the receiver observes the same set of channel outputs at the same time instants. Blue circle markers denote decoding of a message containing reading $0$, whereas red square markers report decoding a message with reading $1$. The dash-dotted line shows the source stationary entropy $H(X)$. Results obtained for $q=0.005$,  $m = 392$ ($\rho\simeq0.05$), $R=10$m, $q=0.01$, $\alpha=0.02$, $\zeta=10^{-4}$, $\peras = 0.1$.}
    \label{fig:entropy-evolution}
\end{figure*}

\begin{example}
\label{ex:mindful_symm}
    A first example of the evolution of $\condent(y^n)=H(\Mcn\given Y^n=y^n)$ over time is reported in Fig.~\ref{subfig:ex_1}, considering a symmetric source and using an FS receiver. The snapshot was obtained for a coverage range of $\radMax=50$ m, with $\alpha=0.02$, $\nodes=392$ (corresponding to $\dens\simeq0.05$), $\pTx=10^{-4}$, $\peras=0.1$. In the plot, the reception of a message containing reading $0$ is highlighted with blue circle markers, whereas reception of a $1$ is indicated with red square markers.  In the early phase of the example, after collecting several $0$ messages in a row, the receiver has a strong belief about the current state of the source, with $\mathsf P[\Mcn=0\given Y^n = y^n]$ reaching $0.948$. The subsequent reception of a (possibly unreliable) reading with value $1$ then leads to a rapid spike in the level of uncertainty, which then progressively reduces as further messages corroborate a change in the state of the source. The effect of taking all history into account is apparent even with a qualitative comparison with the simpler, more regular trend discussed in Example~\ref{ex:forgetful} for the forgetful receiver. On the one hand, uncertainty can suddenly rise when a message contradicting a strong belief arrives. On the other hand, the reset values of entropy upon decoding are not always the same as in the forgetful case, but rather strongly depend on the past collected messages.
\end{example}

\begin{example}
\label{ex:mindful_asymm}
    Figs.\ \ref{subfig:ex_2} and \ref{subfig:ex_3} present the time evolution of $\condent(y^n)$ for a FS receiver, when tracking an asymmetric source, with $\eta=5$ and $\eta=20$, respectively. The same parameters discussed in Example \ref{ex:mindful_symm} were used and the receiver experiences the same sequence of incoming readings at the same time. Clearly, as $\eta$ increases, the stationary source entropy $H(X)$, i.e., the uncertainty at which the receiver tends when not receiving updates for long time, decreases (horizontal dashed lines in the plots).
    Note that, while $\condent(y^n)$ for the tracked binary source can never exceed the stationary entropy $H(X) = 1$ in the symmetric case, the instantaneous uncertainty of the receiver can exceed the value $H(X)<1$ for larger values of $\eta$. By definition, the entropy is limited to $1$ regardless of the value of $\eta$, as $X_n$ is a binary variable.
    Let us first consider the case of a slight asymmetry in the source, as in Fig.\ \ref{subfig:ex_2}. Here, the reception of a single reading with value $1$ leads to a high uncertainty, and while the second consecutive packet reduces the uncertainty, the overall entropy is higher than in the case with $\eta=1$, since state $1$ is generally less likely. We can observe that, after receiving packets reporting the less frequent value, the entropy gradually grows as the probability that the system is actually in that state decreases over time. When the probability flips, i.e., there is a higher chance that the system is in state $0$, the entropy reaches $1$ and then decreases, asymptotically converging from above to $H(X)$. On the other hand, when the reported value is $0$, the entropy immediately after a packet is much lower, gradually increasing with time and converging to $H(X)$ from below. If the asymmetry grows significantly, as in Fig. \ref{subfig:ex_3}, which shows the trend for $\eta=20$, even multiple consecutive readings confirming that the system is in the less frequent state lead to an entropy close to $1$, as the low probability of being in state $1$ leads the FS receiver to maintain a high uncertainty. The convergence towards $H(X)$ is then much faster from state $1$ and correspondingly slower from state $0$.
    
    For completeness, \figr\ref{subfig:ex_4} reports the uncertainty $\condent_f(\lastrcn,\agen)$ of a forgetful receiver exposed to the same sequence of observations, considering a source with $\eta=5$. As discussed in Example \ref{ex:forgetful}, the effect of ignoring previously received readings results in resetting $\condent_f(\lastrcn,\agen)$ to the uncertainty of the incoming message. In the case of an asymmetric process \Mcn, this corresponds to a higher value when receiving the less likely state $1$ or a lower value when collecting a $0$.\footnote{Similar trends were noted in \cite{Liew22_TIT} for the UoI, which, as discussed, coincides with the studied metric for the forgetful receiver in the special case of perfectly reliable readings.} A comparison with \figr\ref{subfig:ex_2} provides a visual intuition of the sub-optimality of the forgetful approach.
\end{example}

\subsection{Numerical results and discussion}

Unless otherwise stated, all the parameters for the results are set as specified in Sec. \ref{sec:results_forgetful}. The uncertainty \entFun\ of the FS receiver was estimated via Monte Carlo simulations, taking the mean across multiple runs of $h(y^n)$ over sufficiently long intervals to estimate the limiting behavior in \eqref{eq:ent_mindful_def}.
\begin{figure}
    \centering
    \includegraphics[width=0.63\columnwidth]{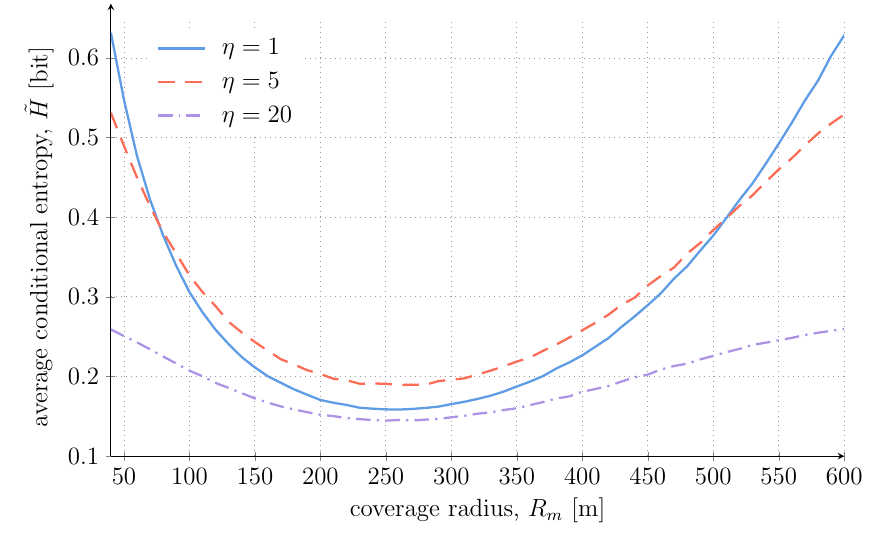}
    \caption{Average conditional uncertainty of the FS receiver, $\entFun$, vs coverage range \radMax. A binary source is monitored, considering different asymmetry factors. $\dens\simeq0.05$, $R=10$~m, $q=0.005$, $\alpha=0.02$, $\zeta=10^{-4}$, $\peras = 0.1$.}
    \label{fig:HvsRad_minfdul}
\end{figure}

As a first result, we report in \figr\ref{fig:HvsRad_minfdul} \entFun\ against the coverage range \radMax. Three types of sources are considered: symmetric ($\eta=1$, solid curve), slightly asymmetric ($\eta=5$, dashed curve), and highly asymmetric ($\eta=20$, dash-dotted curve). The fundamental trade-off driven by range selection that emerged in the forgetful case is confirmed even when all collected messages over time are taken into account. Large uncertainty is experienced both for a small coverage, due to sporadic updates, and for exceedingly large coverage due to an increase in collisions and to the reception of less reliable readings.
An interesting trend emerges considering the behavior under different values of $\eta$. In intermediate coverage ranges, i.e., when the receiver is fed with a good mix of reliable and less reliable messages, a relatively low asymmetry in the source leads to higher values of \entFun\ in comparison to the tracking of a symmetric process. The result can be explained referring to the discussion carried out in Examples \ref{ex:mindful_symm} and \ref{ex:mindful_asymm}, as the retrieval of sequences of possibly erroneous messages reporting a visit to the less likely state hinders an accurate belief at the receiver. When asymmetry increases further (e.g., $\eta=20$), the trend is more than counterbalanced by the higher predictability of the source behavior (see, e.g., \figr\ref{subfig:ex_3}). Notably, asymmetry is always beneficial in terms of \entFun\ for very low and very large coverage radii. In the former case, updates are highly reliable, so that uncertainty is reset to low values with every incoming message, and even a slight reduction of the source stationary entropy at which $\condent(y^n)$ tends to converge without new readings is useful. In the latter, most messages are poorly informative and \entFun\ will eventually converge to the source entropy $H(X)$.

\begin{figure}
    \centering
    \subfloat[Optimal coverage range.]{
        \includegraphics[width=.47\columnwidth]{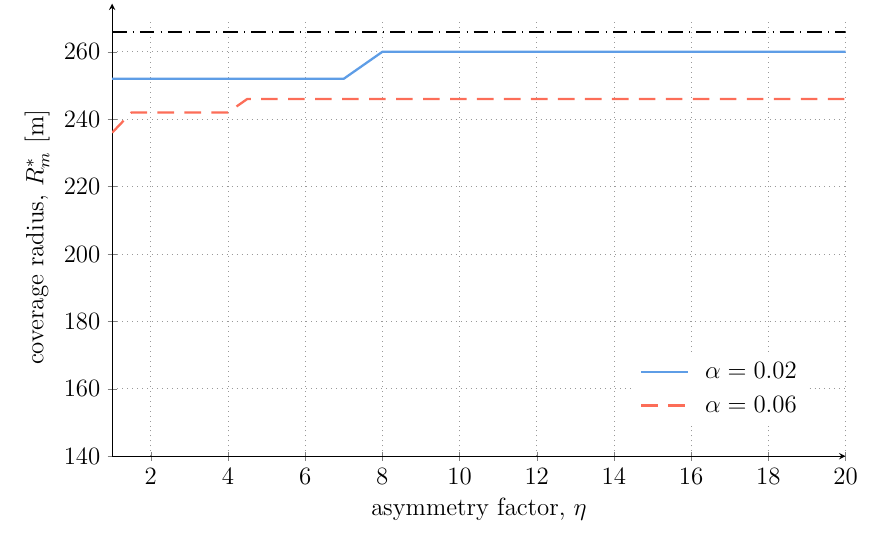}}
    \subfloat[Minimum conditional entropy.\label{subfig:forg_ent}]{
        \includegraphics[width=.47\columnwidth]{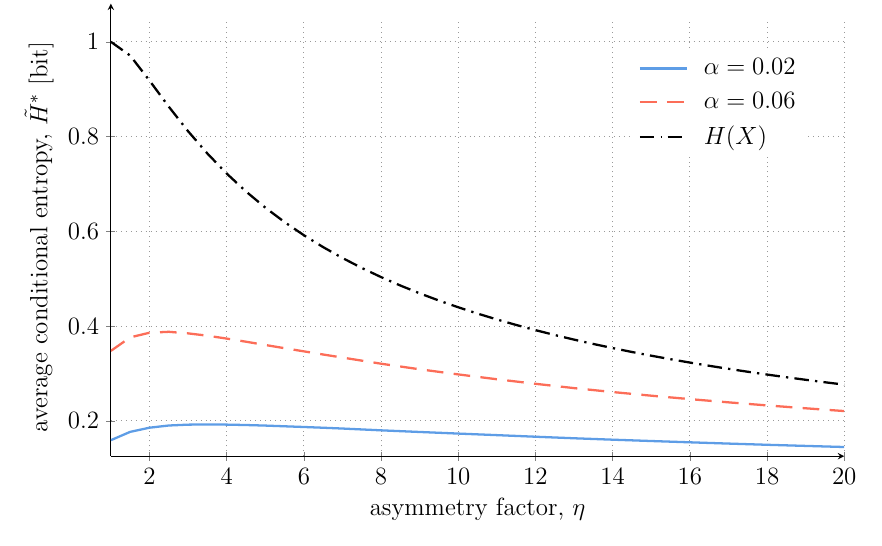}}
    \caption{Optimal coverage range $\radMax^*$ and minimum average conditional entropy $\entFun^*$ obtained with $\radMax^*$ vs asymmetry factor $\eta$ for a FS receiver. The horizontal line corresponds to coverage for minimum AoI in (a) and to the source stationary entropy in (b).}
    \label{fig:forgetful-hmm-eta}
\end{figure}

These aspects are further explored in \figr\ref{fig:forgetful-hmm-eta}, showing the optimal coverage range and the corresponding minimum conditional entropy \entFun\ against the asymmetry factor, considering an FS receiver and two values of the exponent $\alpha$ of the reading reliability function. The role played by $\eta$ is apparent, with asymmetry being detrimental up until $\eta \simeq 15$ for the set of parameters under study and $\alpha=0.02$. As expected, a larger $\alpha$ leads to higher uncertainty, and the optimal system configuration is achieved for a smaller $\radMax^*$, sacrificing the frequency of updates to favor their reliability. It is also interesting to compare the results with the trends reported in \figr\ref{fig:optHandRadius} for the forgetful case. The optimal coverage range for the FS approach is larger and converges to the value that minimizes \ac{AoI} (i.e., ${\sim}265$~m) more quickly as $\eta$ increases. Overall, the possibility to leverage beliefs on the full history of received updates allows to better cope with potentially less reliable readings, and benefits from an increase in the frequency of update receptions. In all cases, extending the coverage radius beyond $(\pi\rho\pTx(1-\peras))^{-1/2}$ is detrimental, as it would trigger an overloaded channel (i.e., load larger than $1$ pkt/slot) and result in a reduction in the number of received messages. In general, the importance of properly tuning the coverage range is also evident for an FS receiver, although the optimal setting is closer to what would be obtained considering \ac{AoI} alone. This can be seen in \figr\ref{fig:HvsRad_minfdul}, noting how uncertainty deteriorates rather quickly when \radMax\ moves away from its optimal value. We can also note that the FS receiver increases the gain over the stationary entropy even for higher values of $\alpha$, as shown in \figr\ref{subfig:forg_ent}: this further confirms that the FS receiver is able to exploit even unreliable information, reducing the entropy more than the forgetful receiver even for low values of $\eta$.

\begin{figure*}
    \centering
    \subfloat[Symmetric source ($\eta=1$).]{
        \includegraphics[width=.47\columnwidth]{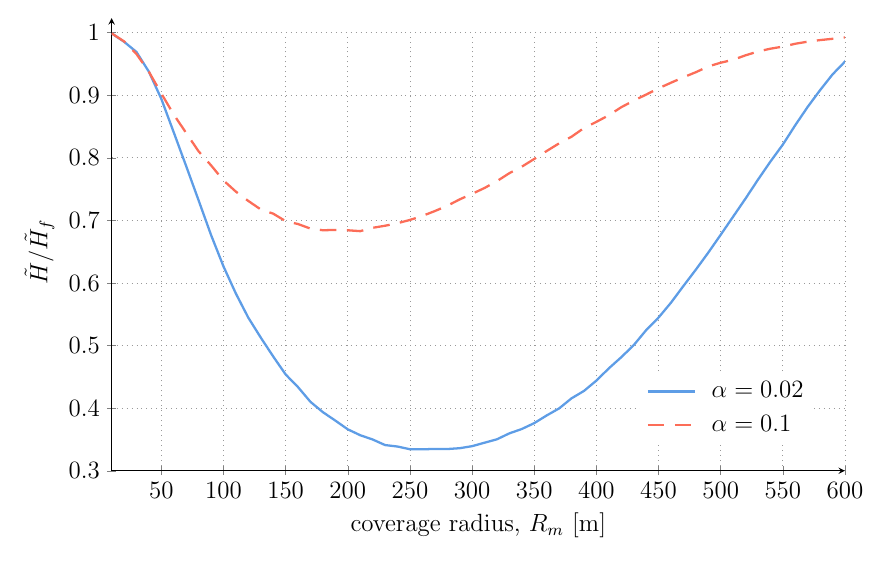}}
    \hspace*{.3em}
    \subfloat[Fixed reliability ($\alpha=0.02$).\label{subfig:fixed_rel}]{
        \includegraphics[width=.47\columnwidth]{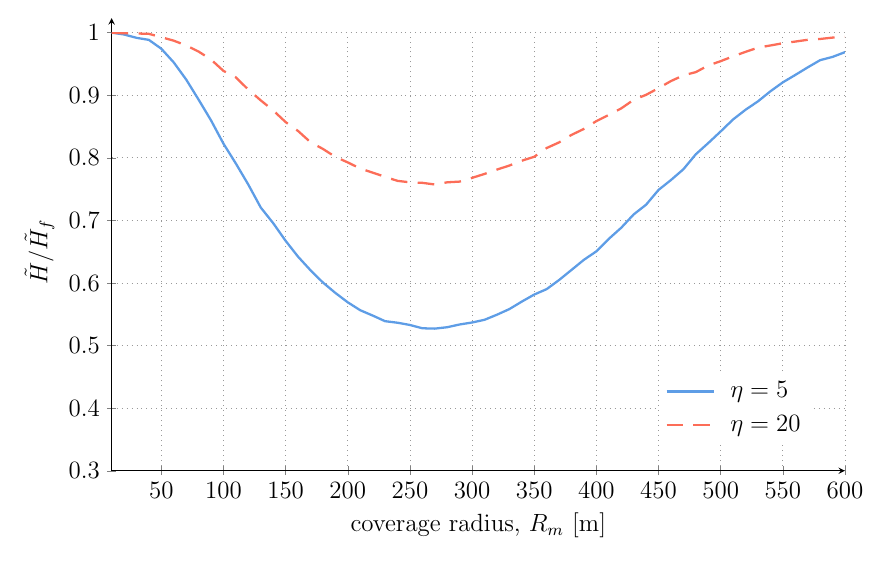}}
    \caption{Ratio of the average conditional entropy of the FS and of the forgetful estimator, i.e., $\entFun/\entFun_f$, vs coverage range \radMax. In (a), a symmetric source is tracked ($\eta=1$) and different values of $\alpha$ are considered. In (b), $\alpha=0.02$, and two values of the asymmetry factor $\eta$ for the source are shown.}
    \label{fig:ratio}
\end{figure*}

A direct comparison of the performance achieved with the two types of receiver is shown in \figr\ref{fig:ratio}, reporting the ratio $\entFun/\entFun_f$ versus the coverage radius \radMax, and considering different sets of parameters. 
In all cases, both receivers exhibit the same behavior for very small and very large coverage areas. The former condition is formally captured in Remark~\ref{rem:ideal}, as all readings are perfectly reliable for a single coverage area. As \radMax\ diverges, the channel is dominated by collisions and the uncertainty converges to $H(X)$. On the other hand, significant improvements are obtained for intermediate coverage, with a reduction of the uncertainty up to a factor $\sim 3$ attained by leveraging the whole set of collected observations $Y^n$. The most notable gains are experienced when $\alpha$ increases, as the FS receiver manages to better handle incoming unreliable readings by relying on past messages. Smaller, yet still noticeable gaps can be observed when reliability degrades less sharply with distance or in the presence of highly asymmetric sources, as shown in \figr\ref{subfig:fixed_rel}.

\begin{figure}
    \centering
    \includegraphics[width=0.63\columnwidth]{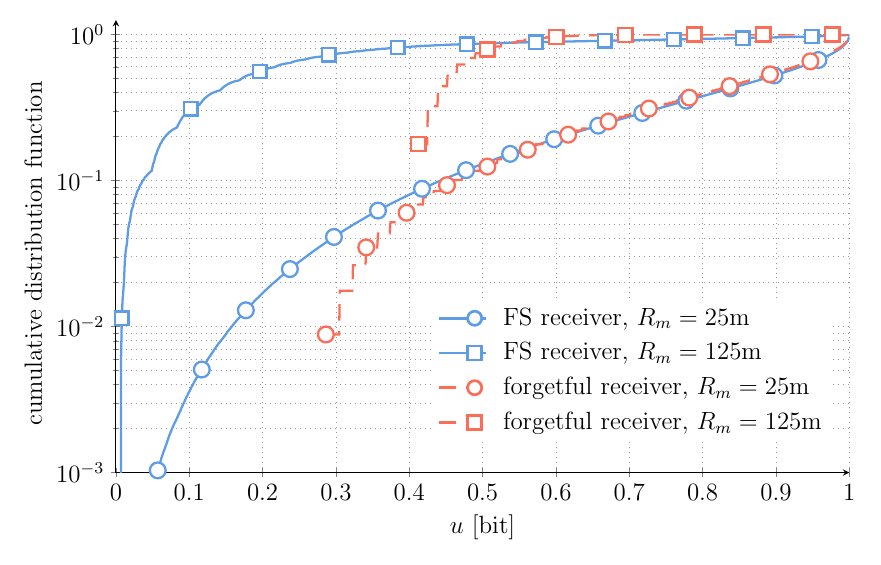}
    \caption{Cumulative distribution function of $\condent_f(\lastrcn,\agen)$ for the forgetful receiver, and of $\condent(y^n)$ for the FS receiver. Different values of coverage range are considered, tracking a symmetric source.}
    \label{fig:cdf}
\end{figure}

Further insight on the behavior of the two receivers is offered in \figr\ref{fig:cdf}, reporting the \ac{CDF} of the receiver uncertainty at a generic point in time, i.e., of $\condent_f(\lastrcn,\agen)$ (dashed lines) and $\condent(y^n)$ (solid lines). A symmetric source is considered ($\eta=1$), and results are shown for two different coverage radii: a relatively small one ($25$ m, circle markers) and a larger one ($125$ m, square markers). For the forgetful receiver, the metric can be obtained from \eqref{eq:condEntropy_def}, \eqref{eq:pYnStat}-\eqref{eq:pAgenStat} as   
\begin{align}
\mathsf P\big[ \condent(\lastrcn,\agen) \leq u \big] = \sum_{\substack{(\lastrcn,\agen) \text{ s.t.}\\ \condent(\lastrcn,\agen)\leq u}} p(\lastrcn,\agen).
\end{align}
In the FS case, the empirical \ac{CDF} was derived via Monte Carlo simulations, tracking the values of $\condent(y^n)$ over multiple episodes.
The plot pinpoints how the uncertainty of the forgetful approach is bounded from below by the entropy of the state of the source conditioned on the value of an incoming message, determined by the distribution $\pb(\mcn\given\lastrcn)$ computed in \eqref{eq:pXnGivenYnReset}. The more unreliable readings can be, i.e., the larger the radius, the higher the uncertainty of the receiver upon decoding (see Example~\ref{ex:forgetful}). The situation changes when the whole sequence $Y^n$ is accounted for. In this case, the receiver can experience a much lower uncertainty, e.g., by receiving a sequence of updates, reinforcing the belief of the source about being in a specific state (Example \ref{ex:mindful_symm}). The result is also interesting from a system design standpoint, providing guidelines on the choice of the receiver complexity and on the coverage radius when not only the average uncertainty but also the probability of not exceeding a possibly critical value for the application of interest is to be considered.

Finally, we analyze the interplay of the maximum coverage radius $\radMax$, the accuracy decay parameter $\alpha$, the transmission probability $\pTx$, and the average number of nodes per unit area $\dens$. Fig.~\ref{fig:additional_results}(a) shows the entropy of the two receivers as a function of the coverage radius under different settings: we note that increasing the number of nodes per unit area to $\dens'=2\dens$ allows the receiver to obtain more accurate information if $\radMax$ is set to a low value, as there are more sensors providing accurate readings. However, the larger number of nodes increases the risk of collision: setting a relatively high $\radMax$ results in a much higher entropy due to the lower success probability. On the other hand, 
operating with $2\dens$ and $\pTx/2$ results in almost exactly the same performance as with the original values, because doubling the number of nodes in each area and halving their transmission probability leads to approximately the same collective behavior. This also highlights a system design trade-off: a denser sensor network may be more expensive to deploy, but the lifetime of individual nodes will be longer, as their transmission rates, and, thus, their energy expenditure, will decrease accordingly without any performance loss in terms of the receiver entropy.

The effect of the reliability parameter $\alpha$ is shown in Fig.~\ref{fig:additional_results}(b), considering a fixed coverage radius and transmission probability. Following Remark~\ref{rem:ideal}, we see that the two receiver models are identical for $\alpha=0$. The FS receiver is then better able to exploit a larger number of sensors, while the entropy of the forgetful receiver quickly grows with $\alpha$ independently of the number of nodes in the area. However, even relatively low values of $\alpha$ substantially degrade the entropy of both receiver types.

\begin{figure*}
    \centering
    \subfloat[Average entropy as a function of \radMax, with $\alpha=0.02$.\label{fig:entropy-load}]{
        \includegraphics[width=.47\columnwidth]{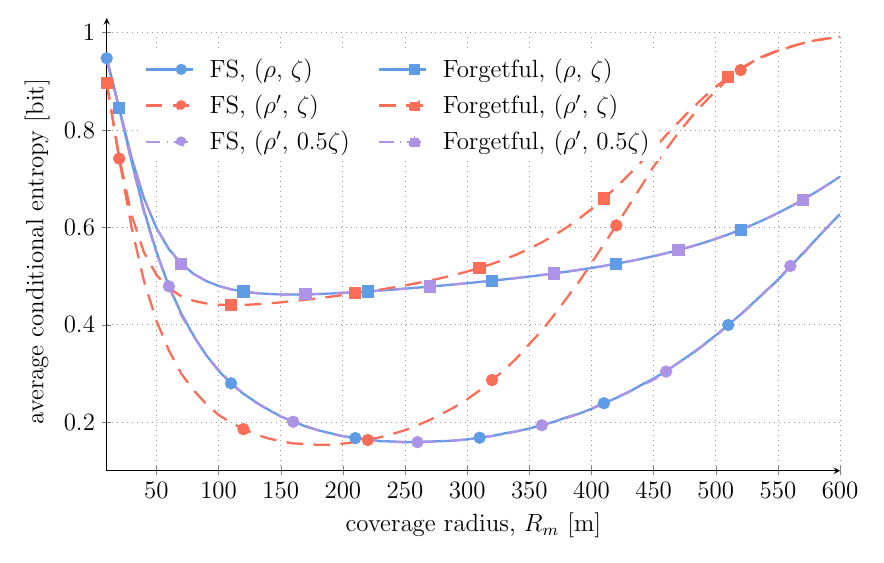}}
    \hspace*{.3em}
    \subfloat[Average entropy as a function of the reliability parameter $\alpha$ with $\nodes=2454$ and $\nodes'=2\nodes=4908$, corresponding to $\dens$ and $\dens'$ nodes per unit area, respectively, with $\radMax=125$m.\label{fig:entropy-reliability}]{
        \includegraphics[width=.47\columnwidth]{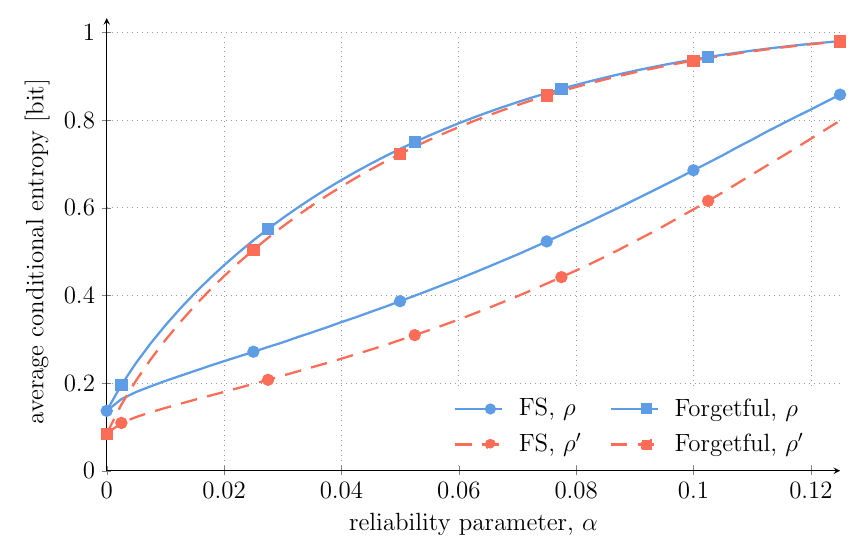}}
    \caption{Average uncertainty at the FS (\entFun) and forgetful ($\entFun_f$) receiver, with $\dens\simeq0.05$ and $\dens'=2\dens\simeq0.1$, $\zeta=10^{-4}$, $R=10$ m, $\eta=1$, $q=0.005$, $\peras=0.1$.}
    \label{fig:additional_results}
\end{figure*}

\section{System Optimization}\label{sec:optim}

The results discussed so far depended on two simplifying assumptions: \textit{(1)}, the receiver cannot rely on any geographical information for received packets, only accounting for the average reliability level of incoming readings; \textit{(2)}, all devices employ the same transmission probability $\pTx$, irrespective of their location, and thus on the reliability of their updates. 
Both aspects trigger relevant questions on whether system performance could be optimized by relaxing the working hypotheses. In the first case,  uncertainty might be reduced by properly weighting the reliability of the content of a decoded packet based on its source's accuracy. The problem is also relevant from a system design standpoint, as acquiring and reporting geographical information entails additional complexity at the device side (e.g., requiring GNSS capabilities) as well as overhead to notify the position of the transmitter in a sent message. Similarly, more complex protocols that adapt the transmission probability to the position of the source could be considered, accounting for the different accuracy as well as for the larger number of devices that populate farther regions. 

\subsection{The role of location awareness}\label{sec:locationAware}

Let us first consider the impact of position-aware nodes. To study this, we extend our model assuming that each device knows the region $\area_i$ it belongs to, and piggybacks this information whenever sending an update. Consequently, the
observation $Y_n$ has alphabet $(\mathcal X \times \mathcal K) \cup \{\bigast\}$, and the analytical framework presented for the two types of receiver we consider can be readily extended to derive the conditional distribution on the current state of the source given the available knowledge. Specifically, we have:
\begin{itemize}
    \item \emph{Forgetful receiver:} in this case, the uncertainty at time $n$ is driven not only by the latest received reading, $\lastRcn$, and the time elapsed since its decoding, $\Agen$, but also by the region from which it originated, $\Distn$. Following the same approach discussed in Sec. \ref{sec:forgetful_analysis}, upon reception of a packet at time $n_0$, we get
    \begin{align}
        \pb(\mc_{n_0}\given \lastrc_{n_0},\dist_{n_0}) =  \frac{\pb(\lastrc_{n_0},\dist_{n_0} \given \mc_{n_0}) \, \pi_{\mc_{n_0}}}{\sum\limits_{\mc_{n_0}^{\prime}\in \mathcal X} \pb(\lastrc_{n_0},\dist_{n_0}\given \mc_{n_0}^{\prime}) \, \pi_{\mc_{n_0}^{\prime}}}\,.
    \end{align}
    Note that, as the location information is available to the receiver, the condition is on both $\lastrc_{n_0}$ and $\dist_{n_0}$, while the original forgetful receiver update in~\eqref{eq:pXnGivenYnReset} only uses the former, computing the emission probability $\pb(\lastrc_{n_0}\given\mc_{n_0})$ by applying the law of total probability in~\eqref{eq:pYnGivenXnReset}. In this case, the emission probability is
    \begin{align}
        \pb(\lastrc_{n_0},\dist_{n_0} \given \mc_{n_0}) = \pb(\lastrc_{n_0}\given\mc_{n_0},\dist_{n_0}) \, \pb(\dist_{n_0}),
        \label{eq:pWnDnGivenXn}
    \end{align}
    where $\pb(\lastrc_{n_0}\given\mc_{n_0},\dist_{n_0})$ and $\pb(\dist_{n_0})$ were derived in \eqref{eq:pYnGivenDnXn} and \eqref{eq:pDnGivenXn}, respectively. The distribution for any $n=n_0+\agen>0$ between the current and next decoding follows by taking the $\agen$-step evolution of the process via the transition matrix \TransMat, allowing us to compute $\condentforg(\lastrcn,\agen,\distn)$ for a generic time slot $n$. The average uncertainty at the receiver, $\entFun_f$, can in turn be evaluated by taking the expectation of this quantity over the joint distribution     \begin{align}
        p(\lastrcn,\agen,\distn) = p(\agen) \sum_{\mc_{n_0}\in\mathcal X} \pb(\lastrc_{n_0},\dist_{n_0} \given \mc_{n_0}) \,\pi_{\mc_{n_0}},
    \end{align}
    where we leveraged again the independence of the AoI from the content and origin of the last received message.
    \item \emph{FS receiver:} the \ac{HMM} characterizing the relationship between source state and channel output can again be applied, taking into account the fact that at every slot the receiver either observes $\bigast$ (collision or idle slot), or a pair $(\lastrcn,\distn)$ in the event of successful decoding. The emission probability in the latter case is given in \eqref{eq:pWnDnGivenXn}, allowing us to obtain the conditional probability via the standard \ac{HMM} forward recursion. 
\end{itemize}

\begin{figure*}
    \centering
    \subfloat[$\alpha=0.02$.]{
        \includegraphics[width=.47\columnwidth]{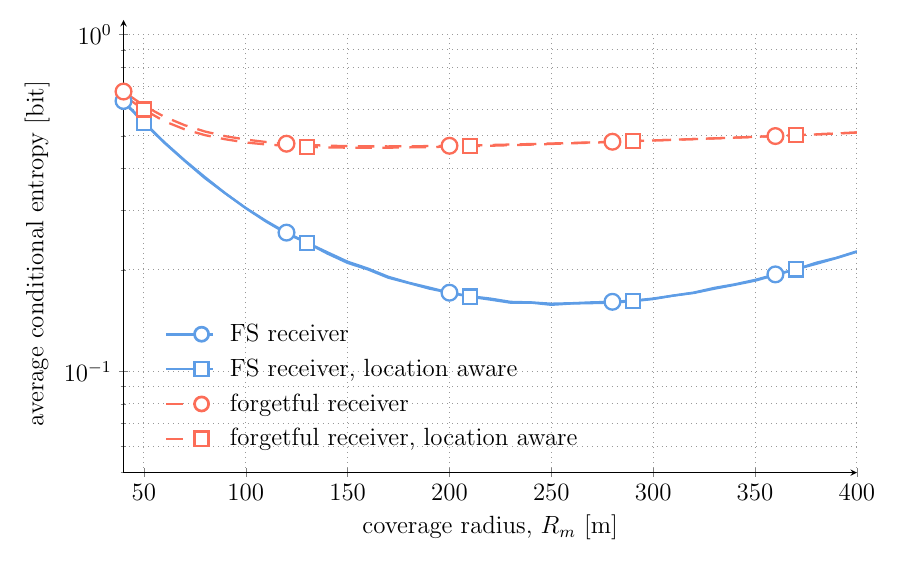}}
    \hspace*{.3em}
    \subfloat[$\alpha=0.1$.]{
        \includegraphics[width=.47\columnwidth]{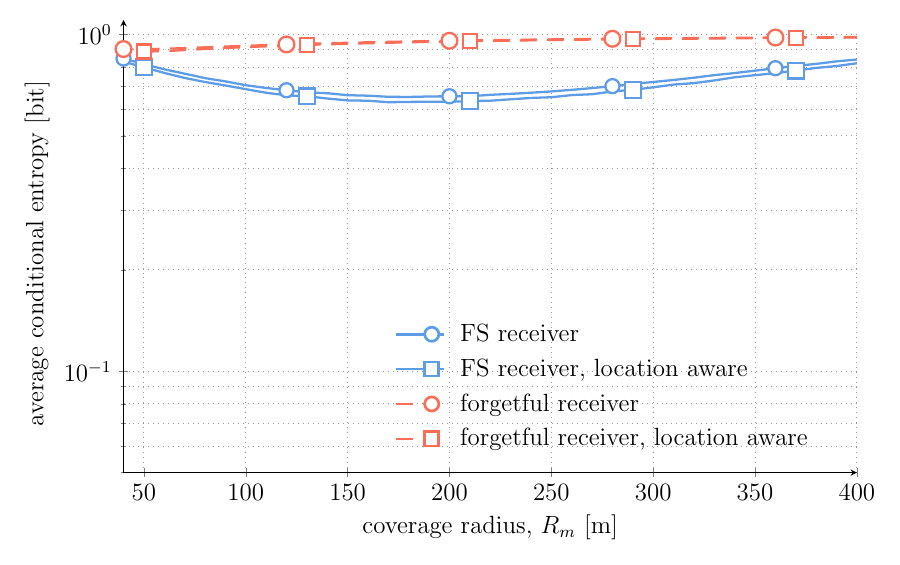}}
    \caption{Average conditional entropy for a forgetful ($\entFun_f$) and a FS (\entFun) receiver, reported against the coverage radius \radMax. Circle markers denote the performance achieved when the receiver has no knowledge about the position of the transmitters, whereas square markers refer to the case of location being piggybacked into sent messages. In (a), the reliability parameter $\alpha$ is set to $0.02$, whereas (b) refers to the case $\alpha=0.1$. A symmetric source is monitored.}
    \label{fig:location-aware}
\end{figure*}

The results of this analysis are shown in \figr\ref{fig:location-aware}, where we report the average conditional entropy for the forgetful (dashed lines) and FS (solid lines) receivers, considering the case without (circle markers) and with (square markers) location awareness. A symmetric source was considered, and the groups of curves refer to two values for the reading reliability parameter $\alpha$.

The first and key outcome that can be extracted is that piggybacking location information on the updates has a very limited impact on the overall uncertainty at the receiver, both in the forgetful and in the FS case. The result is non-obvious and provides a relevant design take-away, supporting the relevance of simpler IoT systems in which devices need not be aware of their position. From this standpoint, the modeling assumptions introduced in \secr\ref{sec:sysModel} prove to be sound, leading to a broad applicability of the main setting considered in our framework. A detailed look at \figr\ref{fig:location-aware} also reveals that a slight improvement (up to $1.96\%$) can be noticed for the forgetful receiver already for $\alpha=0.02$. In the FS case, instead, the role of weighing the reliability of incoming updates based on the distance of their sender only emerges when $\alpha$ is slightly larger, with an overall reduction in \entFun\ up to $3.73\%$ for $\alpha=0.1$. This relatively small effect for both receiver types is due to the relative frequency of low-quality updates, which is quite high due to the spatial distribution of nodes: when receiving sequences of low-quality updates, the two receivers still have a high entropy, and the knowledge of the location of the higher-quality updates has a small effect due to their low relative frequency.

\subsection{Transmission probability optimization}

To understand whether an adaptation of the transmission probability can be beneficial, consider an extension of the model in which all nodes within the same region $\area_i$ access the channel with probability $\pTx_i$. We thus obtain a vector $\bm \pTx = [\pTx_0, \dots, \pTx_{K-1}]$, and can find the optimal configuration $\bm \pTx^*$ by solving the problem
\begin{align} \label{eq:problem1}
    \bm\pTx^*=\underset{\bm \pTx\in[0,\pTx_{\max}]^K}{\operatorname{argmin}} & \quad \entFun \\
    \text{s.t.}& \quad 0 \leq \pTx_i \leq \pTx_{\max}\ \forall i,
\end{align}
where $\pTx_{\max}$ is the maximum transmission rate for a node, e.g., due to duty cycle or other operational constraints. 
To this aim, we define the set of threshold-based solutions $\mathcal{T}$. A solution $\bm\pTx$ belongs to $\mathcal{T}$ if and only if it respects the following condition:
\begin{equation}
\begin{aligned}
    \bm\pTx\in\mathcal{T}\iff\exists K^*\in\{0,\ldots,K-1\}:&\pTx_i=\pTx_{\max}\ \forall i<K^*\\ &\wedge \pTx_j=0\ \forall j>K^*.
\end{aligned}
\end{equation} We also define the accuracy $\bar{\lambda}(\bm\pTx)$ of a solution, i.e., the probability of a packet received when using solution $\bm\pTx$ reporting the correct value of the state, as
\begin{equation}
    \bar{\lambda}(\bm\pTx)=\pb(\lastrc_{n_0}\given \mc_{n_0})\mathbbm{1}(\mc_{n_0},\lastrc_{n_0}).
\end{equation}
We then state two Lemmas, which will be useful in determining the optimal solution.

\begin{lemma}\label{lm:eq_success}
    Consider any solution $\bm\pTx$. If $\nodes_i$ is a non-decreasing function of $i$, i.e., there are more nodes in outer areas, there exists a threshold-based solution $\bm\pTx'\in\mathcal{T}$ with an equal success probability, i.e.,
    \begin{equation}
        \nodes_j\geq\nodes_i\ \forall j>i\implies\forall\bm\pTx\ \exists\bm\pTx'\in\mathcal{T}:\ps(\bm\pTx')=\ps(\bm\pTx).
    \end{equation}
\end{lemma}
\begin{IEEEproof}
    The proof of the Lemma is provided as Appendix~\ref{app:lemma3}. 
\end{IEEEproof}

\begin{lemma}\label{lm:accuracy}
If we consider the set $\mathcal{L}(\ps)=\{\bm\pTx:\ps(\bm\pTx)=\ps\}$, the solution 
\begin{equation}
    \bm\pTx^*(\ps)=\underset{\bm \pTx\in\mathcal{L}(\ps)}{\operatorname{argmax}}\ \ \bar{\lambda}(\bm\pTx)
\end{equation}
belongs to the set $\mathcal{T}$ of threshold-based policies if $\rad$, $\powLaw$, and $\pTx_{\max}$ satisfy 
\begin{equation}
    \frac{1}{(1+\rad)^{\powLaw}}\geq\frac{1-\pTx_{\max}}{(1+2\rad)^{\powLaw}}+\pTx_{\max}.\label{eq:cond_j1}
\end{equation}
\end{lemma}
\begin{IEEEproof}
    The proof of the Lemma is provided as Appendix~\ref{app:lemma4}. 
\end{IEEEproof}

We then use the Lemmas to prove the following result.

\begin{theorem} \label{th:opt_pTx}
If the number of nodes in each zone $\nodes_i$ is a non-decreasing function of $i$, and $\frac{1}{(1+\rad)^{\powLaw}}\geq\frac{1-\pTx_{\max}}{(1+2\rad)^{\powLaw}}+\pTx_{\max}$, there exists an optimal activation probability vector $\bm \pTx^*$ presenting a threshold structure, i.e., $\bm\pTx^*\in\mathcal{T}$. In other words, the intersection between the set of optimal solutions to problem~\eqref{eq:problem1} and the set of threshold-based solutions $\mathcal{T}$ is not empty.
\end{theorem}
\begin{IEEEproof}
We start with a trivial observation: any solution that generates a load higher than $1$, i.e., with $\sum_{i=0}^{K-1}\nodes_i\pTx_i>1$, cannot be optimal, as we can always find a solution with the same ratio of packets from each zone and a higher success probability by reducing the load accordingly. We then concentrate only on solutions with a load lower than $1$. 
Consider now any pair of solutions $\bm\pTx$ and $\bm\pTx'$, with the same success probability, i.e., $\ps(\bm\pTx)=\ps(\bm\pTx')$. It is immediate to show that, if $\bar\lambda(\bm\pTx') \geq \bar\lambda(\bm\pTx)$, i.e., received packets have a higher probability of containing the correct state, the entropy of the receiver is never higher when the nodes transmit according to $\bm\pTx'$ than when using $\bm\pTx$.

The theorem then straightforwardly follows from the two lemmas: as Lemma~\ref{lm:eq_success} states that we can always construct a threshold-based solution with the same throughput as any other solution, and Lemma~\ref{lm:accuracy} states that the highest-accuracy solution with a given success probability must be threshold-based, there is always a threshold-based solution which leads to an entropy at least as low as any other solution. The set of optimal solutions must then always include at least a threshold-based one.
\end{IEEEproof}

Theorem \ref{th:opt_pTx} offers another relevant hint for system design, proving how optimization of the transmission probability based on the positions of the nodes does not play a key role. This is especially relevant, since implementing a per-zone channel access vector $\bm \pTx^*$ would require once more some form of location awareness with respect to the tracked process for all nodes and entail further protocol overhead to distribute the transmission parameters. Both aspects can be critical in IoT systems, where the simplicity of devices and protocols is paramount. 

\begin{figure}
    \centering
    \includegraphics[width=0.63\columnwidth]{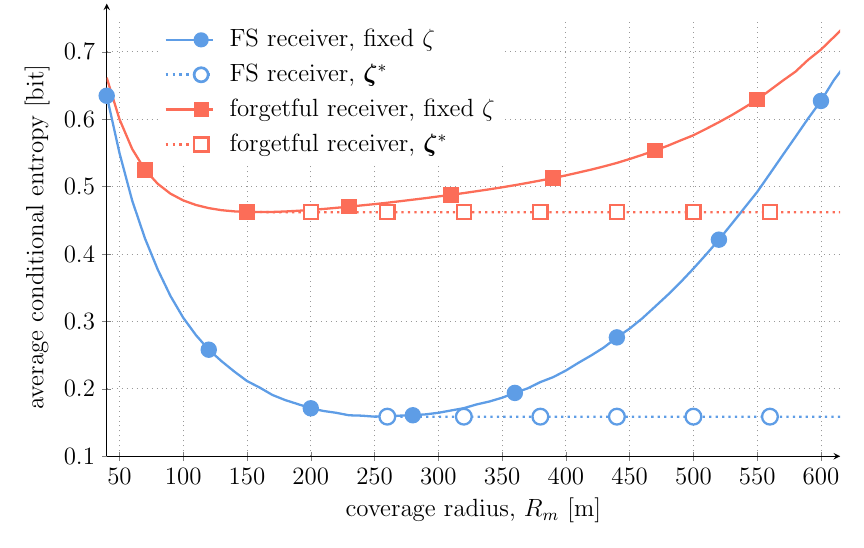}
    \caption{Average conditional entropy vs coverage range \radMax, considering a symmetric source. Solid lines report the behavior of an FS and a forgetful receiver when all nodes within coverage transmit with the same probability $\pTx=10^{-4}$. Dashed lines denote the performance of the two receivers when the transmission probabilities for each coverage zone are obtained by solving the optimization problem \eqref{eq:problem1}.}
    \label{fig:opt_pTx}
\end{figure}
Moreover, the result confirms that the modeling assumptions introduced in Sec. \ref{sec:sysModel} are valid: by simply setting an appropriate coverage range and allowing nodes in the area to transmit as dictated by their duty cycles, the system achieves a nearly optimal solution, as the only difference from $\bm \pTx^*$ is given by (at most) the most external active region. This is clearly shown in \figr\ref{fig:opt_pTx}, reporting \entFun\ against \radMax\ when tracking a symmetric source, and considering the case in which nodes access the channel with probability $\pTx=\pTx_{\max}=10^{-4}$ (solid lines) or with $\bm \pTx^*$ (dashed lines). The latter has been obtained by solving the optimization problem in \eqref{eq:problem1} via standard numerical tools. The performance of both FS (circle markers) and forgetful (square markers) receivers is reported. Up until $\radMax = \radMax^*$, the solutions coincide, as the optimal approach foresees all nodes to transmit with the same probability. Conversely, for larger radii, the uncertainty obtained for a constant $\pTx$ increases due to the reception of unreliable updates, and the optimal solution is simply to switch off any transmission from outside $\radMax^*$. In this perspective, the simpler modeling approach tackled in this paper proves to be especially valuable in determining the optimal coverage and the minimum attainable uncertainty. Also note how the effect of keeping the nodes in region $K^*$ active, with $\pTx_{K^*} < \pTx_{\max}$ in this case, as dictated by the optimal structure of Theorem \ref{th:opt_pTx}, does not play a significant role.

\begin{figure}
    \centering
    \includegraphics[width=0.63\linewidth]{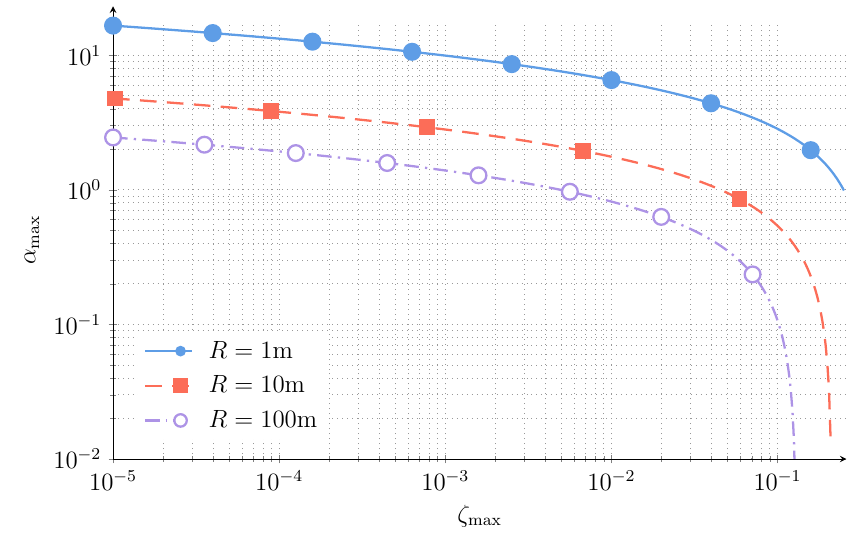}
    \caption{Maximum value of $\powLaw$ under which Theorem~\ref{th:opt_pTx} holds as a function of $\pTx_{\max}$ for different values of $\rad$.}
    \label{fig:alpha_condition}
\end{figure}

To conclude, we also observe that, while the theorem sets some conditions on the values of $\powLaw$, $\rad$, and $\pTx_{\max}$, practical system parameters 
are well within these bounds. 
This is illustrated in Fig.~\ref{fig:alpha_condition}, where, for any value of the granularity $\rad=\radMax/K$, i.e., the width of each annular region within the coverage area, the region below the corresponding curve gives all $\alpha$ values for which the threshold-based solution is proven to be optimal. Decreasing the duty cycle or the granularity of the zone division allows for larger values of $\powLaw$. All results in the simulations we show in this work fall within the condition in the theorem.

\section{Monitoring of an $M$-State Process}
\label{sec:multiState}

To conclude our study, we tackle the more general case of tracking a Markov source with an arbitrary number of states $M$ and alphabet $\mathcal X = \{0,\dots,M-1\}$. We assume once more that the nodes are not location aware and that all use a common transmission probability \pTx.

The one-step transition probabilities of the considered source model are illustrated in \figr\ref{subfig:Mstate}, providing a direct extension of the symmetric source case studied so far and leading to a stationary distribution $\pi_i = 1/M$, $\forall i \in \mathcal X$. In the remainder, we will focus solely on the FS receiver case, whose uncertainty \entFun\ can be captured extending the \ac{HMM} that describes the relationship between \Mcn\ and the channel outputs \Rcn.
Specifically, for any distribution $p(z_n\given \mcn, i)$, $i\in \mathcal K$, describing the probability that a device will produce reading $Z_n$ given the current state of the tracked process and its distance from it, the emission probability of the \ac{HMM} can be obtained as
\begin{align}
    p(\rcn\given \mcn) =
    \begin{cases}
    \ps \sum_{\distn\in\mathcal K} p(z_n \given \mcn,\distn) \, \pb(\distn), &\quad \rcn\in\mathcal X;\\
    1 - \ps, & \quad \rcn = \bigast,
    \end{cases}
\end{align}
where $\pb(\distn)$ is as derived in \eqref{eq:pDnGivenXn}.

\begin{figure}
\centering
    \subfloat[$M$-state source.\label{subfig:Mstate}]{
    \includegraphics[width=0.2\linewidth]{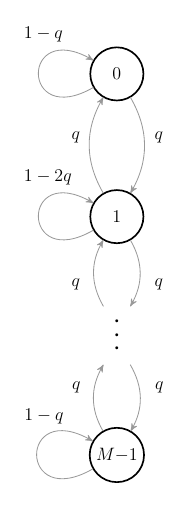}}
    \hspace{3em}
    \subfloat[Example of $p(z_n\given\mcn,i)$.\label{subfig:M_ex}]{
    \includegraphics[width=0.47\linewidth]{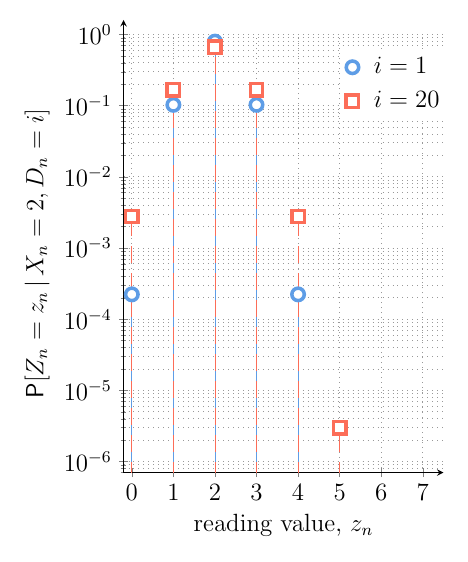}}
    \caption{A diagram of the $M$-state Markov source, with transition probabilities, and an example of the PMF of the observation produced by a sensor in area $\mathcal A_i$ given the state of the source, i.e., $p(z_n\given \mcn,i)$. The result is obtained assuming that $\mathcal{X}=\{0,\ldots,7\}$ and the true source state is $X_n=2$, and considering the model in \eqref{eq:gaussian}-\eqref{eq:sigma2}. To determine $\sigma_i^2$, the reliability exponent of the $2$-state model was set to $\alpha=0.02$.}
    \label{fig:source_and_PMF}
\end{figure}

In the remainder, we focus on a specific example for the reading reliability function, considering a discrete Gaussian distribution
\begin{align}
    p(z_n \given \mcn, i) =  \frac{e^{-(z_n-\mcn)^2/(2\sigma_i^2)}} {\sum_{z'_n\in\mathcal X} e^{-(z'_n-\mcn)^2/(2\sigma_i^2)}},
    \label{eq:gaussian}
\end{align}
where the parameter $\sigma_i^2$ controls how likely it is to obtain readings farther away from the true value. An example of the distribution is reported in \figr\ref{subfig:M_ex}, for $M=8$, showing the PMF of $Z_n$ conditioned on having the source in state $X_n  = 2$ when the reading is produced by a node very close to the center ($i=2$, solid line) or farther away ($i=10$,  dashed line). In order to have a relevant comparison with the results presented earlier, we set the parameter $\sigma_i^2$ so that, for $M=2$, the probability that a device in area $\mathcal A_i$ produces a correct reading corresponds to the one considered in \secr\ref{sec:sysModel}. This is obtained if $p(z_n\given \mcn,i) = \lambda_i$ when $z_n = \mcn$. Imposing this condition within \eqref{eq:gaussian} leads, after straightforward manipulations, to
\begin{align}
    \sigma_i^2 = \left(2 \ln \left(\frac{\lambda_i}{1-\lambda_i} \right)\right)^{-1}.
    \label{eq:sigma2}
\end{align}

Using this setting, we report in \figr\ref{fig:multi_state} the average uncertainty \entFun\ against the coverage range for a source with $M=2$ (solid line), $M=4$ (dashed line), and $M=8$ (dash-dotted line) states. In all cases, the transition probability $q$ is set to $5\cdot 10^{-3}$ and $\alpha=0.02$.
As expected, a larger number of source states leads to a higher uncertainty, with the effect being especially visible for very low or large coverage, where $\entFun \to H(X) = -\sum_{x\in\mathcal X} \pi_x \log_2 \pi_x$.
More interestingly, the fundamental trends and trade-offs discussed for the two-state case are confirmed, validating the approach of considering a simpler system model and allowing us to broaden the messages provided in the paper.

\begin{figure}
    \centering
    \includegraphics[width=0.63\columnwidth]{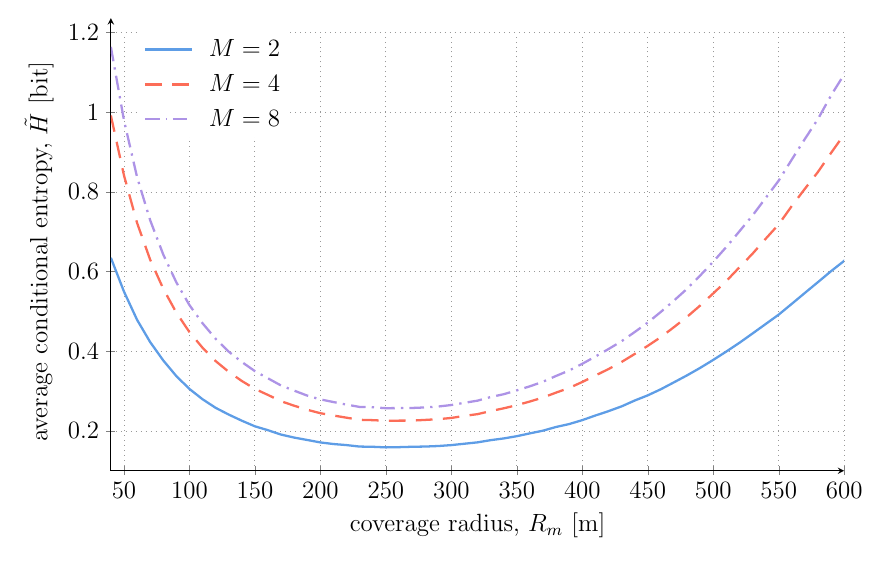}
    \caption{Average conditional entropy of the FS receiver when tracking a multi-state source with $M=8$. The source transitions with probability $q=5\cdot10^{-3}$, and the reliability parameters for the produced readings are set following \eqref{eq:gaussian} and \eqref{eq:sigma2}. To determine $\sigma_i^2$, the reliability exponent of the 2-state model was set to $\alpha=0.02$.}
    \label{fig:multi_state}
\end{figure}

\section{Concluding Remarks}
\label{sec:concs}

The freshness of information is a key metric for the development of 6G massive access, and its study is moving beyond the simple definition of \ac{AoI} and toward more expressive metrics that directly target application performance. At the same time, the temporal aspect of distributed \ac{IoT} observations, which has been extensively studied, needs to be complemented by a spatial component, integrating readings from sensors that may be farther from the event of interest and thus less reliable, but still useful.

In this work, we characterized the uncertainty of the receiver by studying its entropy in such a spatio-temporal observation system, showing its behavior under state asymmetry and multi-state Markov chains, and provide some general insights for optimization. Further, we showed that system optimization can be carried out with simple policies, which can be implemented in limited-hardware devices, with negligible performance loss.

Future avenues of research include the characterization of a wider set of spatio-temporal systems, considering multiple points of interest and structural distance metrics based on graph theory instead of simple Euclidean distance, in order to fully represent modern \ac{IoT} systems while still providing a solid theoretical characterization that can be used for further optimization.
This can be extended to a wider consideration of the system: we can imagine a vehicular traffic estimation problem measuring flow at crossings, i.e., nodes on a graph, in which entropy can capture the structure and system dynamics holistically. Furthermore, our research aims at triggering novel contributions on the topic of spatio-temporal freshness, e.g., exploring it under different channel models and taking into account the specific physical layer and protocol aspects of practical IoT systems.
\appendix

\subsection{On the definition of $\tilde H$} \label{app:lemma1}
We prove that $\tilde H$ is well defined by showing that $\lim\nolimits_{n\to\infty}H(X_n\given Y^n)$ exists. To this aim, consider the joint process $(X_n,Y_n)$. We have $p(x_n,y_n\given x^{n-1},y^{n-1}){=}p(x_n,y_n\given x_{n-1})$, since $X_n$ is Markovian, and $Y_n$ is only driven by the current value of the source. Accordingly, $(X_n,Y_n)$ is a finite-state Markov chain, which is time-homogeneous (as $X_n$ is homogeneous and the observation probabilities over the channel do not change over time). The chain is also irreducible. To prove this, it is sufficient to show that the one-step transition probability $p(x_{n{+}1},y_{n{+}1}\given x_n,y_n)$ between any two states is strictly positive \cite{TaylorKarlin1998}. We then observe that $p(x_{n{+}1},y_{n{+}1}\given x_n,y_n)$ can be factorized as $p(y_{n{+}1}\given x_{n{+}1}) p(x_{n{+}1}\given x_n)$, and that both factors are positive. For the transition probabilities of the tracked process, $p(x_{n{+}1}\given x_n)>0$ by definition, as we assume $q>0$ (see Sec. \ref{sec:sysModel}) and do not consider pathological cases in which the monitored Markov chain stays indefinitely in one state. In turn, $p(y_{n{+}1}\given x_{n{+}1})>0$ follows from \eqref{eq:emission_prob}, recalling that $p_s>0$ (as clarified in \eqref{eq:psucc} since the erasure probability is $p_e<1$ and $\zeta\in(0,1)$) and that $\pb_{\lastRcn\given\Mcn}$ is also strictly positive from \eqref{eq:pYnGivenXnReset_law}. Incidentally, we also remark that, in the case $K=1$, i.e., coverage only consists of the innermost region, all updates are reliable, and the process only visits states $(x_n,y_n)$ with $y_n\in\{x_n,\bigast\}$. In this special case, one would thus consider a reduced chain including only those states, which is irreducible following the same reasoning. Finally, $(X_n,Y_n)$ is also aperiodic, since, e.g., state $(0,0)$ has a strictly positive self-transition probability from the arguments above. 
These conditions ensure that the process is ergodic and a stationary distribution exists. In the remainder, we will assume that the initial distribution is the stationary one, so that the process is stationary. 

We now note that
   \begin{align*}
       H(X_n\given Y^n) &= H(X_n\given Y_n,\dots,Y_0)\\
       &\leq H(X_n\given Y_n, \dots, Y_1)\\
       &\stackrel{(a)}{=} H(X_{n-1}\given Y_{n-1}, \dots, Y_0),
    \end{align*}
    where ($a$) follows from the stationarity of $(X_n,Y_n)$. The sequence $H(X_n\given Y^n)$ is thus non-increasing and non-negative, proving its convergence to a limit. 

\subsection{Proof of Lemma 1} \label{app:lemma2}
Consider first $\Agen$. The process evolves from $\Delta_{n-1}$ as $\Delta_n{=}\Delta_{n-1}+1$ with probability $1{-}\ps$ (no packet received) or $\Delta_n{=}0$ otherwise. We thus have a countable state Markov chain, which can easily seen to be irreducible. The chain is also aperiodic (as the self-transition probability for state $0$ is $\ps>0$) and positive recurrent (since the return time to state $0$ is a geometric r.v. of parameter \ps, with mean value $1/\ps>0$). The process is then ergodic, and we will assume its initial distribution to be equal to the stationary distribution. Ergodicity can be inferred also for  $(X_n,W_n)$, following same the arguments provided in Appendix 1 for $(X_n,Y_n)$. Observing that $\Agen$ is independent of $(X_n,W_n)$, the joint process $(X_n,W_n,\Agen)$ is also a Markov chain, irreducible and aperiodic. Its stationary distribution is the product of the stationary distributions of $(X_n,W_n)$ and \Agen, and is thus proper, proving the ergodicity and 
stationarity of the chain. Accordingly, $H(X_n\given W_n,\Delta_n)$ is the same for all $n$. The limit thus exists and can be computed using the stationary distribution of $(W_n,\Delta_n)$ as in the rightmost equality of the Lemma statement.

\subsection{Proof of Lemma 2}\label{app:lemma3}
The success probability of a setting in which an arbitrary number $n$ of nodes transmit with probability $\pTx_{\max}$ is
\begin{equation}
    \ps(n)=n\pTx_{\max}(1-\peras)(1-\pTx_{\max}(1-\peras))^{n-1}\,.
\end{equation}
In the following, we define symbol $\pTx_{\text{eff}}=(1-\peras)\pTx_{\max}$ for the sake of brevity.

We can then find the value $n^*$ that solves $\ps(n)=\ps(\bm\pTx)$. The equation does not have a closed-form solution, but it can be solved numerically, and $n^*$ exists if the number of nodes $\nodes_i$ is non-decreasing, as the maximum throughput for $n$ transmitting nodes is limited to $(1-\peras)/n$. If we take the second derivative of $\ps(n)$ with respect to $n$, we get
\begin{equation}
    \frac{\partial^2\ps(n)}{\partial n^2}=\pTx_{\text{eff}}(1-\pTx_{\text{eff}})^{n-1}\log(1-\pTx_{\text{eff}})(2+n\log(1-\pTx_{\text{eff}})).
\end{equation}
We know that $\pTx_{\text{eff}}(1-\pTx_{\text{eff}})^{n-1}$ is always positive.
On the other hand, $\log(1-\pTx_{\text{eff}})$ is always negative, as $\pTx_{\text{eff}}<1$. In order for the second derivative to be negative, we then need to set $n\log(1-\pTx_{\text{eff}})>-2$.

Observe that, to maintain a load lower than $1$, we must have $n\pTx_{\text{eff}}\leq1$. Any solution with a load higher than $1$ has a lower success probability than solutions with a lower load, so this does not reduce the validity of the Lemma. We can thus set a more restrictive condition:
\begin{equation}
 n\log(1-\pTx_{\text{eff}})>-2n\pTx_{\text{eff}} \implies n\log(1-\pTx_{\text{eff}})>-2.
\end{equation}
We then note that
\begin{equation}
  n\log(1-\pTx_{\text{eff}})>-2 \iff\frac{\partial^2\ps(n)}{\partial n^2}<0.
\end{equation}
Removing the factor $n$ on each side, we have $\log(1-\pTx_{\text{eff}})>-2\pTx_{\text{eff}}$, which is true for $\pTx_{\text{eff}}<0.796$. As $n>1$, this is a less restrictive condition than $n\pTx_{\text{eff}}\leq1$, and the second derivative is always negative in the interval of interest.

The function is thus concave, and we can find the value of $n^*$ with the bisection method.
We then determine $K^*$ as:
\begin{equation}
    K^*=1+\sup\left\{k\in\{0,\ldots,K-2\}:\sum_{i=0}^k\nodes_i<n^*\right\}.
\end{equation}
We denote the number of nodes in the zones before $K^*$ as $\nodes_{b}=\sum_{i=0}^{K^*-1}\nodes_i$. If the inequality is strict, we then need to find the correct value $\pTx_{K^*}$, which is the solution of
\begin{equation}
\begin{aligned}
  \ps(\bm\pTx)=&\left[\frac{(1-\peras)\pTx_{K^*}\nodes_{K^*}}{1-(1-\peras)\pTx_{K^*}}+\frac{\pTx_{\text{eff}}\nodes_{b}}{1-\pTx_{\text{eff}}}\right](1-\pTx_{\text{eff}})^{\nodes_{b}}\\
  &\times(1-(1-\peras)\pTx_{K^*})^{\nodes_{K^*}}.
\end{aligned}
\end{equation}
We can find this solution numerically using the bisection method, as the success probability monotonically increases with $\pTx_{K^*}$ as long as $(\nodes_{-K^*}+\nodes_{K^*})\pTx_{\text{eff}}\leq1$.

\subsection{Proof of Lemma 3}\label{app:lemma4}

Consider a generic solution $\bm\pTx\notin\mathcal{T}$. If $\bm\pTx$ is not threshold-based, we need to prove that a threshold-based solution $\bm\pTx'$ exists, with $\ps(\bm\pTx')=\ps(\bm\pTx)$, and $\bm\pTx'$ has a higher accuracy. We can do so through a water-filling approach, by proving the existence of a swap in $\bm\pTx$ towards the inner regions that leaves $\ps(\bm\pTx)$ unchanged but improves the accuracy. Thus, we can reach the desired threshold-based $\bm\pTx'$ at the end of a necessarily finite-step procedure where, as long as two values $\pTx_j$ and $\pTx_k$ are both inside $(0,\pTx_{\max})$, with $j<k$, a transfer is made from $\pTx_k$ to $\pTx_j$ to either saturate $\pTx_j$ to $\pTx_{\max}$ or decrease $\pTx_k$ to $0$.

Formally, we define a series of intermediate solutions $\bm\pTx^{(n)}$, with $\bm\pTx^{(0)}=\bm\pTx$. We consider zones $j=\inf\left\{i:\pTx^{(n)}_i<\pTx_{\max}\right\}$, and $k=\sup\left\{i:\pTx^{(n)}_i>0\right\}$. If $j\geq k$, the $n$-th solution is threshold-based, and the procedure is over. Otherwise, we construct a solution $\bm\pTx^{(n+1)}$, with $\pTx^{(n+1)}_i=\pTx^{(n)}_i\ \forall i\notin\{j,k\}$ and $\ps\left(\bm\pTx^{(n+1)}\right)=\ps\left(\bm\pTx^{(n)}\right)$, such that $\pTx^{(n+1)}_j=\pTx_{\max}\vee\pTx^{(n+1)}_k=0$. This is always possible thanks to Lemma~\ref{lm:eq_success}, and, due to the finite dimension of vector $\bm\pTx$, the sequence must converge to a threshold-based policy in a finite number of steps. Lemma~\ref{lm:accuracy} must then hold if every intermediate step improves the accuracy, i.e., if $\bar{\lambda}\left(\bm\pTx^{(n+1)}\right)\geq\bar{\lambda}\left(\bm\pTx^{(n)}\right)\ \forall n$. This can be proven to hold under the relatively mild condition of the hypothesis and after a long sequence of algebraic steps. In the following, we consider the case in which $\peras=0$ for the sake of brevity, without loss of generality. The Lemma holds under the same conditions for other values of $\peras$, after a change of variable from $\pTx_i$ to $z_i=(1-\peras)\pTx_i$. As this is a linear transformation that preserves the sign, and $(1-\peras)\pTx_{\max}\leq\pTx_{\max}$, all the following inequalities are true for $\peras>0$.

We start by modeling all zones other than $j$ and $k$ as a combined zone $c$, and denote the probability of all nodes in combined zone $c$ being silent as $s_c$, the probability of a single one transmitting as $a_c$, and the expected accuracy of packets received from the combined zone $c$ is $\bar{\lambda}_c$, all of which stay constant in the water-filling transfer from $k$ to $j$. The success probability $\ps\left(\bm\pTx^{(n)}\right)$ can then be expressed as
\begin{equation}
\begin{aligned}
   \ps\left(\bm\pTx^{(n)}\right)=&\left[a_c+s_c\left(\frac{\pTx^{(n)}_j}{1-\pTx^{(n)}_j}+\frac{\pTx^{(n)}_k}{1-\pTx^{(n)}_k}\right)\right]\\
   &\times \left(1-\pTx^{(n)}_j\right)^{\nodes_j}\left(1-\pTx^{(n)}_k\right)^{\nodes_k}
\end{aligned}
\end{equation}
and condition $\ps\left(\bm\pTx^{(n+1)}\right)=\ps\left(\bm\pTx^{(n)}\right)$ becomes
\begin{equation}\label{eq:equal_success_gamma}
\begin{aligned}
    \frac{a_c+s_c\left(\frac{\nodes_j\pTx^{(n)}_j}{1-\pTx^{(n)}_j}+\frac{\nodes_k\pTx^{(n)}_k}{1-\pTx^{(n)}_k}\right)}{\left(1-\pTx^{(n+1)}_j\right)^{\nodes_j}\left(1-\pTx^{(n+1)}_k\right)^{\nodes_k}}=\\
    \frac{a_c+s_c\left(\frac{\nodes_j\pTx^{(n+1)}_j}{1-\pTx^{(n+1)}_j}+\frac{\nodes_k\pTx^{(n+1)}_k}{1-\pTx^{(n+1)}_k}\right)}{\left(1-\pTx^{(n)}_j\right)^{\nodes_j}\left(1-\pTx^{(n)}_k\right)^{\nodes_k}}.
\end{aligned}
\end{equation}
Similarly, the the overall expected accuracy $\bar{\lambda}\left(\bm\pTx^{(n)}\right)$ is
\begin{equation}
\begin{aligned}
  \bar{\lambda}\left(\bm\pTx^{(n)}\right)=&\frac{\left(1-\pTx^{(n)}_j\right)^{\nodes_j}\left(1-\pTx^{(n)}_k\right)^{\nodes_k}}{\ps\left(\bm\pTx^{(n)}\right)}\\ &\times\left[\bar{\lambda}_c a_c+s_c\left(\frac{\nodes_j\pTx^{(n)}_j\lambda_j}{1-\pTx^{(n)}_j}+\frac{\nodes_k\pTx^{(n)}_k\lambda_k}{1-\pTx^{(n)}_k}\right)\right]
\end{aligned}
\end{equation}
and we have $\bar{\lambda}\left(\bm\pTx^{(n+1)}\right)\geq\bar{\lambda}(\bm\pTx^{(n)})$ 
if and only if
\begin{equation}
\begin{aligned}
   \frac{\left[\bar{\lambda}_c a_c+s_c\left(\frac{\nodes_j\pTx^{(n+1)}_j\lambda_j}{1-\pTx^{(n+1)}_j}+\frac{\nodes_k\pTx^{(n+1)}_k\lambda_k}{1-\pTx^{(n+1)}_k}\right)\right]}{\left(1-\pTx^{(n)}_j\right)^{\nodes_j}\left(1-\pTx^{(n)}_k\right)^{\nodes_k}\ps\left(\bm\pTx^{(n+1)}\right)} \geq\\
   \frac{\left[\bar{\lambda}_c a_c+s_c\left(\frac{\nodes_j\pTx^{(n)}_j\lambda_j}{1-\pTx^{(n)}_j}+\frac{\nodes_k\pTx^{(n)}_k\lambda_k}{1-\pTx^{(n)}_k}\right)\right]}{\left(1-\pTx^{(n+1)}_j\right)^{\nodes_j}\left(1-\pTx^{(n+1)}_k\right)^{\nodes_k}\ps\left(\bm\pTx^{(n)}\right)}.
   \label{eq:constr_accuracy}
\end{aligned}
\end{equation}
For the sake of compactness, we define two auxiliary symbols:
\begin{align}
    \xi_j=&\nodes_j\left(\frac{\pTx^{(n+1)}_j}{1-\pTx^{(n+1)}_j}-\frac{\pTx^{(n)}_j}{1-\pTx^{(n)}_j}\right);\\
    \xi_k=&\nodes_k\left(\frac{\pTx^{(n)}_k}{1-\pTx^{(n)}_k}-\frac{\pTx^{(n+1)}_k}{1-\pTx^{(n+1)}_k}\right).
\end{align}
Due to~\eqref{eq:equal_success_gamma}, \eqref{eq:constr_accuracy} is equivalent to
\begin{equation}
       \frac{a_c}{s_c}\left((\bar{\lambda}_c-\lambda_k)\xi_k-(\bar{\lambda}_c-\lambda_j)\xi_j\right)+(\lambda_j-\lambda_k)\xi_j\xi_k\geq0.
\end{equation}

We then introduce an additional Lemma:
\begin{lemma}\label{lm:swap_cond}
    Under the condition in~\eqref{eq:cond_j1}, we have
    \begin{equation}\label{eq:cond_xi}
\xi_j(\lambda_j-\lambda_k)\geq(\xi_j-\xi_k)(\bar{\lambda}_c-\lambda_k).
\end{equation}
\end{lemma}
\begin{IEEEproof}
    The proof of the Lemma is provided in Appendix~\ref{app:lemma5}.
\end{IEEEproof}
We then note that~\eqref{eq:cond_xi} represents a more restrictive condition, so that~\eqref{eq:constr_accuracy} always holds if~\eqref{eq:cond_xi} is verified.
Lemma~\ref{lm:swap_cond} thus proves that each water-filling transfer improves the expected accuracy of received packets. We thus proved that any individual step in the sequence has $\ps\left(\bm\pTx^{(n+1)}\right)=\ps\left(\bm\pTx^{(n)}\right)$ and $\bar{\lambda}\left(\bm\pTx^{(n+1)}\right)\geq\bar{\lambda}\left(\bm\pTx^{(n)}\right)$. As the sequence must reach a threshold-based policy, the Lemma is proven.

\subsection{Proof of Lemma 4}~\label{app:lemma5}

We first note that, if we have $j=0$, i.e., the innermost zone does not have a transmission probability $\pTx_{\max}$, while external zones have a non-zero transmission probability, the condition is trivially true, as $\lambda_j\geq\bar{\lambda}_c$. We must then handle the two cases with $j=1$ and $j>1$ separately, but first, we introduce another Lemma to simplify the condition in~\eqref{eq:cond_xi}.
\begin{lemma}\label{lm:bound_xi}
    Considering the $n$-th step in the sequence, the ratio $\xi_k/\xi_j$ must respect lower bound
    \begin{equation}\label{eq:xi_bound_value}
        \frac{\xi_k}{\xi_j}\geq1-\pTx_{\max}.
    \end{equation}
\end{lemma}
\begin{IEEEproof}
    The proof of the Lemma is provided in Appendix~\ref{app:lemma6}.
\end{IEEEproof}

In the case in which $j=1$, we consider the loosest possible bound, i.e., $\bar{\lambda}_c\leq1$. We then substitute the bound in~\eqref{eq:xi_bound_value} into~\eqref{eq:cond_xi} to get
\begin{equation}
    \lambda_j\geq(1-\pTx_{\max})\lambda_k+\pTx_{\max}.
\end{equation}
Naturally, the worst case is $k=2$, as it represents the highest possible value of $\lambda_k$. We substitute the definition from~\eqref{eq:powerLaw} to obtain
\begin{equation}\label{eq:alpha_cond}
    \frac{1}{(1+\rad)^{\powLaw}}\geq\frac{1-\pTx_{\max}}{(1+2\rad)^{\powLaw}}+\pTx_{\max},
\end{equation}
which is identical to the condition in~\eqref{eq:cond_j1}.

We now consider the case in which $j>1$, where we can exploit the structure of $\lambda_i$ that we defined in~\eqref{eq:powerLaw} to obtain a tighter upper bound on $\bar{\lambda}_c$. We first note that $c$ contains zones with indices $\{0,\ldots,j-1\}\cup\{j+1,\ldots,k-1\}$. The former set has $\lambda_i>\lambda_j$, while the latter has $\lambda_i<\lambda_j$, as $\lambda$ is a strictly decreasing function. We then note that considering only zones $\{0,\ldots,j-1\}$ yields an upper bound on $\bar{\lambda}_c$. As all these zones have $\pTx^{(n)}_i=\pTx_{\max}$ by the definition of $j$, we know that
\begin{equation}
    \bar{\lambda}_c\leq\sum_{i=0}^{j-1}\frac{\nodes_i\lambda_i}{\sum_{\ell=0}^{j-1}\nodes_{\ell}}=\sum_{i=0}^{j-1}\frac{2i+1}{(1+iR)^{\powLaw}j^2}.
\end{equation}
where we consider the expected number of nodes $\mathbb{E}[\nodes_i]=(2i+1)\nodes K^{-2}$ and $\lambda_i=(1+i\rad)^{-\powLaw}$ by applying the definitions in~\eqref{eq:region_def} and~\eqref{eq:powerLaw}. Under our model, the actual values of $\nodes_j$ and $\nodes_k$ are random, as nodes are uniformly distributed across the area. However, we use the expected value as a proxy for the sake of readability. This does not affect the rest of the proof, as the linearity of the expected value allows us to extend the result directly. We note that $\partial^2\lambda_i/\partial i^2$ is always positive for $\powLaw\geq0$, and $\lambda_i$ is thus convex in $i$. By convexity, we get
\begin{equation}
    \lambda_i\leq\frac{(j-i)\lambda_0+i\lambda_j}{j}
\end{equation}
and $\lambda_0=1$ by the definition in~\eqref{eq:powerLaw}. We then get a looser bound on $\bar{\lambda}_c$, as all values are positive:
\begin{equation}
\begin{aligned}
    \bar{\lambda}_c\leq&\sum_{i=0}^{j-1}\frac{(2i+1)(j-i(1-\lambda_j))}{j^2}\\
    =&\frac{2j^2+3j+1}{6j^2}+\frac{(2j^2-3j+1)\lambda_j}{6j^2}.
\end{aligned}
\end{equation}
We then substitute the bound into~\eqref{eq:cond_xi}:
\begin{equation}
\begin{aligned}
\lambda_j\left[\xi_j(4j^2+3j-1)+\xi_k(2j^2-3j+1)\right]\geq\\
\left[\xi_j-\xi_k(1-\lambda_k)\right](2j^2+3j+1).
\end{aligned}
\end{equation}
As $\lambda_k\leq\lambda_j$, we can further restrict ourselves to
\begin{equation}\label{eq:cond_minlambda}
\lambda_j\left[\xi_j(4j^2+3j-1)-6\xi_kj\right]\geq(\xi_j-\xi_k)(2j^2+3j+1)
\end{equation}
and we have
\begin{equation}
  (\xi_j-\xi_k)(2j^2+3j+1)=\left(1-\frac{\xi_k}{\xi_j}\right)\frac{(2j^2+3j+1)}{4j^2+3j-1-\frac{6\xi_k j}{\xi_j}}.
\end{equation}
If we take the derivative of the right side of the inequality with respect to $\frac{\xi_k}{\xi_j}$, we get
\begin{equation}
    \frac{\partial\left(1-\frac{\xi_k}{\xi_j}\right)\frac{(2j^2+3j+1)}{4j^2+3j-1-\frac{6\xi_k j}{\xi_j}}}{\partial\frac{\xi_k}{\xi_j}}=-\frac{8j^4+6j^3-7j^2-6j-1}{\left(4j^2+3j-1-\frac{6\xi_k j}{\xi_j}\right)^2},
\end{equation}
which is always negative for $j>1$. We then note that setting $\xi_k/\xi_j\geq 1-\pTx_{\max}$ makes the condition in~\eqref{eq:cond_minlambda} always true for $\pTx_{\max}>\frac{1}{2}$.

\subsection{Proof of Lemma 5}\label{app:lemma6}

We first introduce two parameters, $\varphi=\pTx^{(n+1)}_j-\pTx^{(n)}_j$ and $\psi=\left(\pTx^{(n)}_k-\pTx^{(n+1)}_k\right)\varphi^{-1}$. By the definition of the swap operation between zones $j$ and $k$, we must have $\varphi\in\left(0,\pTx_{\max}-\pTx^{(n)}_j\right]$ and $\psi\in\left(0,\pTx^{(n)}_k\varphi^{-1}\right]$. 
Substituting the two values into~\eqref{eq:equal_success_gamma}, and using the relations
\begin{align}
    \frac{\nodes_j\left(\pTx^{(n)}_j+\varphi\right)}{1-\pTx^{(n)}_j-\varphi}=&\frac{\nodes_j\pTx^{(n)}_j}{1-\pTx^{(n)}_j}+\frac{\nodes_j\varphi}{\left(1-\pTx^{(n)}_j\right)\left(1-\pTx^{(n)}_j-\varphi\right)};\\
    \frac{\nodes_k\!\left(\pTx^{(n)}_k-\psi\varphi\right)}{1-\pTx^{(n)}_k+\psi\varphi}=&\frac{\nodes_k\pTx^{(n)}_k}{1-\pTx^{(n)}_k}-\frac{\nodes_k\psi\varphi}{\left(1-\pTx^{(n)}_k\right)\!\left(1-\pTx^{(n)}_k+\psi\varphi\right)},
\end{align}
the condition can be manipulated to get 
\begin{equation}\label{eq:cond_xi_ab}
\begin{aligned}
 \frac{A(0)}{A(\varphi)}=&1+\frac{B\varphi\nodes_j}{\left(1-\pTx^{(n)}_j\right)\left(1-\pTx^{(n)}_j-\varphi\right)}\\
 &-\frac{B\psi\varphi\nodes_k}{\left(1-\pTx^{(n)}_k\right)\left(1-\pTx^{(n)}_k+\psi\varphi\right)},
\end{aligned}
\end{equation}
where we introduced utility symbols
\begin{align}
  A(\varphi)=&\left(1-\pTx^{(n)}_j-\varphi\right)^{\nodes_j}\left((1-\pTx^{(n)}_k+\psi\varphi\right)^{\nodes_k};\\
  B=&s_c\left(a_c+s_c\left(\frac{\nodes_j\pTx^{(n)}_j}{1-\pTx^{(n)}_j}+\frac{\nodes_k\pTx^{(n)}_k}{1-\pTx^{(n)}_k}\right)\right)^{-1}.
\end{align}

We first note that the condition in~\eqref{eq:cond_xi_ab} is trivially true if $\varphi=0$, as the two solutions are identical. 
We then compute the condition for the derivative of the left side of~\eqref{eq:cond_xi_ab} with respect to $\varphi$ to be positive:
\begin{equation}
    \begin{aligned}
        A(\varphi)B\left(\frac{\nodes_j}{\left(1-\pTx^{(n)}_j-\varphi\right)^2}-\frac{\nodes_k\psi}{\left(1-\pTx^{(n)}_k+\psi\varphi\right)^2}\right)\\+\frac{A(0)}{A(\varphi)}\left(\frac{A(\varphi)\nodes_k\psi}{1-\pTx^{(n)}_k+\psi\varphi}-\frac{A(\varphi)\nodes_j}{1-\pTx^{(n)}_j-\varphi}\right)\geq0.
    \end{aligned}
\end{equation}
We then define value $\psi^*$ as follows:
\begin{equation}
    \psi^*=\frac{\nodes_j\left(1-\pTx^{(n)}_k\right)}{\nodes_k\left(1-\pTx^{(n)}_j-\varphi\right)-\nodes_j\varphi}.\label{eq:psi_bound}
\end{equation}
If we set $\psi=\psi^*$, $\ps\left(\bm\pTx^{(n+1)}\right)=\ps\left(\bm\pTx^{(n)}\right)$ for $\varphi=0$ and $\ps\left(\bm\pTx^{(n+1)}\right)\geq\ps\left(\bm\pTx^{(n)}\right)$ for $\varphi>0$. As solution $\bm\pTx^{(n)}$ is not threshold-based, we must have $\varphi>0$. We then remark that the system 
clearly has to operate at load lower than $1$, and thus the success probability must be a strictly decreasing function of $\psi$, which controls the reduction of the transmission probability from zone $k$. This means that, in order to have $\ps\left(\bm\pTx^{(n+1)}\right)=\ps\left(\bm\pTx^{(n)}\right)$, we must have $\psi\geq\psi^*$. We then compute the value of $\xi_k/\xi_j$:
\begin{equation}
    \frac{\xi_k}{\xi_j}=
    \frac{\nodes_k\psi\varphi\left(1-\pTx^{(n)}_j\right)\left(1-\pTx^{(n)}_j-\varphi\right)}{\nodes_j\varphi\left(1-\pTx^{(n)}_k\right)\left(1-\pTx^{(n)}_k+\psi\varphi\right)}.
\end{equation}
As the value is a strictly increasing function of $\psi$, we can substitute $\psi^*$ to obtain a lower bound. After some algebraic manipulations, we get:
\begin{equation}
   \frac{\xi_k}{\xi_j}\geq\frac{1-\pTx^{(n)}_j}{1-\pTx^{(n)}_k}\geq 1-\pTx_{\max}.
\end{equation}

\bibliographystyle{IEEEtran}
\bibliography{IEEEabrv,AoI24,biblio_RandomAccess,bibliography}

@article{testi2025packet,
  author={Testi, Enrico and Paolini, Enrico},
  journal={IEEE Internet Things J.}, 
  title={Packet Collision Probability of Direct-to-Satellite IoT Systems}, 
  year={2025},
  volume={12},
  number={2},
  pages={1843--1855},
}

@ARTICLE{lee2024handover,
  author={Lee, Ju-Hyung and Park, Chanyoung and Park, Soohyun and Molisch, Andreas F.},
  journal={IEEE Trans. Wireless Commun.}, 
  title={Handover Protocol Learning for {LEO} Satellite Networks: Access Delay and Collision Minimization}, 
  year={2024},
  volume={23},
  number={7},
  pages={7624--7637}
}

@article{callebaut2019characterization,
  title={Characterization of {LoRa} point-to-point path loss: Measurement campaigns and modeling considering censored data},
  author={Callebaut, Gilles and Van der Perre, Liesbet},
  journal={IEEE Internet of Things Journal},
  volume={7},
  number={3},
  pages={1910--1918},
  year={2020},
  publisher={IEEE}
}

@article{yavascan2021analysis,
  title={Analysis of slotted {ALOHA} with an age threshold},
  author={Yavascan, Orhan Tahir and Uysal, Elif},
  journal={IEEE J. Sel. Areas Commun.},
  volume={39},
  number={5},
  pages={1456--1470},
  year={2021},
  publisher={IEEE}
}

@article{beltramelli2020lora,
  title={{LoRa} beyond {ALOHA}: An investigation of alternative random access protocols},
  author={Beltramelli, Luca and Mahmood, Aamir and {\"O}sterberg, Patrik and Gidlund, Mikael},
  journal={IEEE Trans. Industrial Informatics},
  volume={17},
  number={5},
  pages={3544--3554},
  year={2020},
  publisher={IEEE}
}

@article{Munari21_TCOM,
  title={Modern random access: an age of information perspective on irregular repetition slotted {ALOHA}},
  author={Munari, Andrea},
  journal=IEEE_J_COM,
  year={2021},
  volume = 69,
  number = 6,
  month = jun,
  pages = {3572-3585}
}

@inproceedings{Kellerer19,
  title={Age-of-information vs. value-of-information scheduling for cellular networked control systems},
  author={Ayan, Onur and Vilgelm, Mikhail and Kl{\"u}gel, Markus and Hirche, Sandra and Kellerer, Wolfgang},
  booktitle={Proc.\ ACM/IEEE ICCPS},
  pages={109--117},
  month={April},
  year={2019}
}

@article{B,
author = "Leonardo Badia",
title = "On the Impact of Correlated Arrivals and Errors on {ARQ} Delay Terms",
journal = IEEE_J_COM,
volume = 57,
number = 2,
pages = {334--338},
month = feb,
year = 2009,
}

@inproceedings{yates2017status,
  title={Status updates over unreliable multiaccess channels},
  author={Yates, Roy D and Kaul, Sanjit K},
  booktitle={Proc.\ IEEE ISIT},
  pages={331--335},
  month={June},
  year={2017},
}

@article{guo2021enabling,
  title={Enabling massive {IoT} toward {6G}: A comprehensive survey},
  author={Guo, Fengxian and Yu, F Richard and Zhang, Heli and Li, Xi and Ji, Hong and Leung, Victor CM},
  journal=IEEE_J_IOT,
  volume={8},
  number={15},
  pages={11891--11915},
  year={2021},
}

@article{Hribar2018_IoT,
  title={Using correlated information to extend device lifetime},
  author={Hribar, Jernej and Costa, Maice and Kaminski, Nicholas and DaSilva, Luiz A},
  journal=IEEE_J_IOT,
  volume={6},
  number={2},
  pages={2439--2448},
  year={2018},
  publisher={IEEE}
}

@inproceedings{Munari25_SPAWC,
    author={A. Munari and F. Chiariotti and L. Badia and P. Popovski},  
    title={Spatio-Temporal Information Freshness for Remote Source Monitoring in {IoT} Systems},
    booktitle={Proc.\ IEEE SPAWC},
    year={2025},
    month={July}
}

@ARTICLE{rabiner1989tutorial,
  author={Rabiner, L.R.},
  journal={Proc. IEEE}, 
  title={A tutorial on hidden {M}arkov models and selected applications in speech recognition}, 
  year={1989},
  volume={77},
  number={2},
  pages={257-286},
}

@STRING{IEEE_J_COM        = "{IEEE} Trans. Commun."}

@STRING{IEEE_J_IT         = "{IEEE} Trans. Inf. Theory"}

@STRING{IEEE_J_IOT        = "{IEEE} Internet Things J."}

@STRING{IEEE_J_NET        = "{IEEE/ACM} Trans. Netw."}

@Book{Cover_Thomas,
  Title                    = {Elements of information theory},
  Author                   = {T. M. Cover and J. A. Thomas},
  Publisher                = {Wiley},
  Year                     = {2006},
  Address                  = {New York},
  Edition                  = {2nd},
}

@book{kovalevsky,
title ={\foreignlanguage{russian}{Читающие автоматы и распознавание образов} [{Character readers and pattern recognition}]},
Author={Valery A. Kovalevsky},
Publisher = {Naukova Dumka},
Year = {1968},
Address = {Kyiv, Ukraine}
}

@Article{Abramson77:PacketBroadcasting,
  author                    = {Norman Abramson},
  title                     = {The Throughput of Packet Broadcasting Channels},
  journal                   = IEEE_J_COM,
  year                      = {1977},
  volume                    = {COM-25},
  number                    = {1},
  pages                     = {117--128},
}

@Article{Ivanonv17:TCOM,
  author                    = {Ivanov, M. and Br\"annstr\"om, F. and Graell i Amat, A. and Popovski, P.},
  title                     = {Broadcast Coded Slotted {ALOHA}: A Finite Frame Length Analysis},
  journal                   = IEEE_J_COM,
  year                      = {2017},
  volume                    = {65},
  number                    = {2},
  pages                     = {651-662},
}

@InProceedings{OnOff2003,
  Title                    = {{The On-Off Fading Channel}},
  Author                   = {Perron, E. and Rezaeian, M. and Grant, A.}, 
  month={June},
  Booktitle                = {Proc. IEEE ISIT},
  Year                     = {2003},
}

@InProceedings{Sun16:PECCSA,
  author =    {Sun, Z. and Xie, Y. and Yuan, J. and Yang, T.},
  title =     {{Coded Slotted ALOHA Schemes for Erasure Channels}},
  booktitle = {Proc. IEEE ICC Workshops},
  year =      {2016},
  month =     {May}
}

@article{chen2019minimum,
  title={Minimum error entropy {Kalman} filter},
  author={Chen, Badong and Dang, Lujuan and Gu, Yuantao and Zheng, Nanning and Pr{\'\i}ncipe, Jos{\'e} C},
  journal={IEEE Trans. Syst., Man, Cybern. Syst.},
  volume={51},
  number={9},
  pages={5819--5829},
  year={2019},
  publisher={IEEE}
}

@article{saridis2002entropy,
  title={Entropy formulation of optimal and adaptive control},
  author={Saridis, George N},
  journal={IEEE Trans. Autom. Control},
  volume={33},
  number={8},
  pages={713--721},
  year={2002},
  publisher={IEEE}
}

@ARTICLE{Cocco23_JSAIT,
  author={Cocco, Giuseppe and Munari, Andrea and Liva, Gianluigi},
  journal={IEEE J. Sel. Areas Inf. Theory}, 
  title={Remote Monitoring of Two-State {Markov sources} via Random Access Channels: An Information Freshness vs. State Estimation Entropy Perspective}, 
  year={2023},
  volume={4},
  number={},
  pages={651-666},
}

@inproceedings{Munari25_ISIT,
  title={On the Uncertainty of a Simple Estimator for Remote Source Monitoring over {ALOHA} Channels},
  author={Andrea Munari},
  booktitle={Proc.\ IEEE ISIT},
  year={2025},
  month={June}
}

@INPROCEEDINGS{Asgari25,
  author={Asgari, Houman and Munari, Andrea and Liva, Gianluigi and Cocco, Giuseppe},
  booktitle={Proc. ITG SCC}, 
  title={Remote Monitoring of Two-State {Markov} Sources in Random Access Channels: Joint Model and State Estimation}, 
  year={2025},
  month={March}
}

@ARTICLE{Liew22_TIT,
  author={Chen, Gongpu and Liew, Soung-Chang and Shao, Yulin},
  journal=IEEE_J_IT, 
  title={Uncertainty-of-Information Scheduling: A Restless Multiarmed Bandit Framework}, 
  year={2022},
  volume={68},
  number={9},
  pages={6151-6173},
  }

@techreport{ericsson2024mobility,
    author = {Ericsson},
    title = {Mobility Report -- {November} 2024},
    institution = {Ericsson},
    year = 2024
}

@ARTICLE{Yates19_TIT,
author={R. D. {Yates} and S. K. {Kaul}},
journal=IEEE_J_IT,
title={The Age of Information: Real-Time Status Updating by Multiple Sources},
year={2019},
volume={65},
number={3}
}

@article{zancanaro2023modeling,
  title={Modeling value of information in remote sensing from correlated sources},
  author={Zancanaro, Alberto and Cisotto, Giulia and Badia, Leonardo},
  journal={Comput. Commun.},
  volume={203},
  pages={289--297},
  year={2023},
  publisher={Elsevier}
}

@Article{Ephremides19_AoII,
  author =  {A. Maatouk and S. Kriouile and M. Assaad and A. Ephremides},
  title =   {The Age of Incorrect Information: a new Performance Metric for Status Updates},
  journal= IEEE_J_NET,
  volume = {28},
  number = {5},
  year  = 2020,
}

@INPROCEEDINGS{Yates17:AoI_SA,
author        ={R. Yates and S. Kaul},
booktitle     ={Proc. IEEE ISIT},
title         ={Status updates over unreliable multiaccess channels},
year          ={2017},
}

@INPROCEEDINGS{Modiano18_AoI,
author      ={R. Talak and S. Karaman and E. Modiano},
booktitle   ={Proc. IEEE SPAWC},
title       ={Distributed Scheduling Algorithms for Optimizing Information Freshness in Wireless Networks},
year        ={2018},
month       ={June},
}

@article{gunduz2023timely,
  title={Timely and massive communication in {6G}: Pragmatics, learning, and inference},
  author={G{\"u}nd{\"u}z, Deniz and Chiariotti, Federico and Huang, Kaibin and Kal{\o}r, Anders E and Kobus, Szymon and Popovski, Petar},
  journal={IEEE BITS Inf. Theory Mag.},
  volume={3},
  number={1},
  pages={27--40},
  year={2023},
  publisher={IEEE}
}

@INPROCEEDINGS{Kaul11_SECON,
author={S. {Kaul} and M. {Gruteser} and V. {Rai} and J. {Kenney}},
booktitle={Proc. IEEE SECON},
title={Minimizing age of information in vehicular networks},
year={2011},
month={June}
}

@article{kalor2024wireless,
  title={Wireless {6G} connectivity for massive number of devices and critical services},
  author={Kal\o{}r, Anders E and Durisi, Giuseppe and Coleri, Sinem and Parkvall, Stefan and Yu, Wei and Mueller, Andreas and Popovski, Petar},
  journal={Proc. IEEE},
  year={2024},
  publisher={IEEE}
}

@article{yang2021spatiotemporal,
  title={Spatiotemporal analysis for age of information in random access networks under last-come first-serve with replacement protocol},
  author={Yang, Howard H and Arafa, Ahmed and Quek, Tony QS and Poor, H Vincent},
  journal={IEEE Trans. Wireless Commun.},
  volume={21},
  number={4},
  pages={2813--2829},
  year={2021},
  publisher={IEEE}
}

@article{tong2022age,
  title={Age-of-information oriented scheduling for multichannel {IoT} systems with correlated sources},
  author={Tong, Jingwen and Fu, Liqun and Han, Zhu},
  journal={IEEE Trans. Wireless Commun.},
  volume={21},
  number={11},
  pages={9775--9790},
  year={2022},
  publisher={IEEE}
}

@article{fidler20242d,
  title={{2D-AoI}: Age-of-Information of Distributed Sensors for Spatio-Temporal Processes},
  author={Fidler, Markus and Gallistl, Flavio and Champati, Jaya Prakash and Widmer, Joerg},
  journal={IEEE Trans. Commun.}, 
  year={2026},
  volume={74},
  pages={645-661}
}

@inproceedings{luo2024minimizing,
  title={Minimizing the age of missed and false alarms in remote estimation of {Markov} sources},
  author={Luo, Jiping and Pappas, Nikolaos},
  booktitle={Proc. 25th ACM MOBIHOC},
  pages={381--386},
  month={October},
  year={2024}
}

@article{salimnejad2023state,
  title={State-aware real-time tracking and remote reconstruction of a {Markov} source},
  author={Salimnejad, Mehrdad and Kountouris, Marios and Pappas, Nikolaos},
  journal={J. Commun. Netw.},
  volume={25},
  number={5},
  pages={657--669},
  year={2023},
  publisher={KICS}
}

@ARTICLE{talli2025pragmatic,
  author={Talli, Pietro and Santi, Edoardo David and Chiariotti, Federico and Soleymani, Touraj and Mason, Federico and Zanella, Andrea and Gündüz, Deniz},
  journal={IEEE J. Sel. Areas Commun.}, 
  title={Pragmatic Communication for Remote Control of Finite-State {Markov} Processes}, 
  year={2025},
  volume={43},
  number={7},
  pages={2589--2603},
}

@article{holm2023goal,
  title={Goal-oriented scheduling in sensor networks with application timing awareness},
  author={Holm, Josefine and Chiariotti, Federico and Kal{\o}r, Anders E and Soret, Beatriz and Pedersen, Torben Bach and Popovski, Petar},
  journal={IEEE Trans. Commun.},
  volume={71},
  number={8},
  pages={4513--4527},
  year={2023},
  publisher={IEEE}
}

@ARTICLE{zakeri2025semantic,
  author={Zakeri, Abolfazl and Moltafet, Mohammad and Codreanu, Marian},
  journal={IEEE Trans. Commun.}, 
  title={Semantic-Aware Sampling and Transmission in Real-Time Tracking Systems: A {POMDP} Approach}, 
  year={2025},
  volume={73},
  number={7},
  pages={4898--4913}
}

@ARTICLE{cosandal2025multi,
  author={Cosandal, Ismail and Akar, Nail and Ulukus, Sennur},
  journal={IEEE Trans. Inf. Theory}, 
  title={Multi-Threshold AoII-Optimum Sampling Policies for Continuous-Time Markov Chain Information Sources}, 
  year={2025},
  volume={71},
  number={9},
  pages={6968--6988},
}

@inproceedings{cosandal2024joint,
  title={Joint age-state belief is all you need: Minimizing {AoII} via pull-based remote estimation},
  author={Cosandal, Ismail and Ulukus, Sennur and Akar, Nail},
  booktitle={Proc. IEEE ICC Workshops},
  month={June},  
  pages={1098--1103},
  year={2025}
}

@inproceedings{liyanaarachchi2025structured,
  title={Structured Estimators: A New Perspective on Information Freshness},
  author={Liyanaarachchi, Sahan and Ulukus, Sennur and Akar, Nail},
  booktitle={Proc. IEEE ITW},
  month=sep,
  year={2025}
}

@ARTICLE{chen2025preempting,
  author={Chen, Yutao and Ephremides, Anthony},
  journal={IEEE Trans. Netw.}, 
  title={Preempting to Minimize Age of Incorrect Information Under Transmission Delay}, 
  year={2026},
  volume={34},
  number={},
  pages={668-680}
}

@INPROCEEDINGS{rezaeianPercom2007,  
author={Rezaeian, Mohammad},  
booktitle={Proc. IEEE PERCOM Workshops},   
title={Sensor Scheduling for Optimal Observability Using Estimation Entropy},   year={2007},
month={March},
}

@article{rezaeianArxiv2006,
  title = {Hidden {Markov} Process: A New Representation, Entropy Rate and Estimation Entropy},
  author={Rezaeian, Mohammad},  
  publisher = {arXiv},
  journal={arXiv preprint cs/0606114},
  year = {2006}
}

@ARTICLE{luo_TIT2009,  
author={Luo, Jun and Guo, Dongning},  
journal={IEEE Trans. Inf. Theory},   
title={On the Entropy Rate of Hidden {Markov} Processes Observed Through Arbitrary Memoryless Channels},   
year={2009},  volume={55},  
number={4},  pages={1460-1467},  doi={10.1109/TIT.2009.2013030}}

@ARTICLE{Liew24_TIT,
  author={Chen, Gongpu and Liew, Soung-Chang},
  journal=IEEE_J_IT, 
  title={An Index Policy for Minimizing the Uncertainty-of-Information of {Markov} Sources}, 
  year={2024},
  volume={70},
  number={1},
  pages={698-721},
  }

@ARTICLE{Feder94_TIT,
  author={Feder, M. and Merhav, N.},
  journal=IEEE_J_IT, 
  title={Relations between entropy and error probability}, 
  year={1994},
  volume={40},
  number={1},
  pages={259-266},
}

@book{TaylorKarlin1998,
  author    = {Taylor, Howard and Karlin, Samuel},
  title     = {An Introduction to Stochastic Modeling},
  edition   = {3rd},
  publisher = {Academic Press},
  address   = {San Diego, CA, USA},
  year      = {1998},
  isbn      = {978-0-12-684887-8}
}

@book{johnson1977urn,
author = "Johnson, N. L. and Kotz, S.", 
title = "Urn Models and Their Application: An Approach to Modern Discrete Probability Theory", 
address = "New York, NY, USA", 
publisher = "John Wiley \& Sons", 
year = 1977,
}

\end{document}